\documentclass[american,sn-mathphys-num,pdflatex]{sn-jnl}
\usepackage[T1]{fontenc}
\usepackage{mathtools}
\usepackage{amsbsy}
\usepackage{amstext}
\usepackage[numbers]{natbib}

\makeatletter
\usepackage{graphicx}%
\usepackage{multirow}%
\usepackage[title]{appendix}%
\usepackage{xcolor}%
\usepackage{textcomp}%
\usepackage{manyfoot}%
\usepackage{booktabs}%
\usepackage{algorithm}%
\usepackage{algorithmicx}%
\usepackage{algpseudocode}%
\usepackage{listings}%

\usepackage{amsmath,amssymb,amsfonts}%
\usepackage{amsthm}%
\usepackage{mathrsfs}%

\newtheorem{theorem}{Theorem}[section]%  meant for continuous numbers

\numberwithin{theorem}{section}
\newtheorem{proposition}[theorem]{Proposition}% 
\newtheorem{corollary}[theorem]{Corollary}% 

\newtheorem{lemma}[theorem]{Lemma}%

\counterwithin*{theorem}{section} 

\newtheorem{problem}{Problem}

\counterwithin*{problem}{section} 
\newtheorem{example}{Example}

\counterwithin*{example}{section}

\theoremstyle{thmstyletwo}%
\newtheorem{remark}{Remark}%

\counterwithin*{remark}{section} 

\theoremstyle{thmstylethree}%
\newtheorem{definition}{Definition}%

\counterwithin*{remark}{definition}

\makeatother

\usepackage{babel}
\begin{document}
\title[Efficient tensor bases for PC]{Efficient tensor bases for pairwise comparisons}

\author*[1]{\fnm{Konrad} \sur{Kułakowski}}\email{konrad.kulakowski@agh.edu.pl}

\author[2]{\fnm{Ryszard} \sur{Smarzewski}}\email{ryszard.smarzewski@wat.edu.pl}
\affil*[1]{\orgname{AGH University of Krakow, Faculty of Electrical Engineering, Automatics,
Computer Science, and Biomedical Engineering }\orgdiv{Applied Computer Science Department}, \orgaddress{\street{al. Mickiewicza 30}, \city{Cracow}, \postcode{30-059},
\country{Poland}}}

\affil[2]{\orgname{Military University of Technology, Faculty of Cybernetics},\orgdiv{Department of Mathematics and Cryptology},
\orgaddress{\street{ul. Kaliskiego 2}, \city{Warsaw}, \postcode{00-908}, \country{Poland}}}
\abstract{In this study, we construct the first orthogonal basis for additively
consistent subspace in pairwise comparisons theory. This construction
is based on our representation of additively consistent best approximations
of skew-symmetric matrices with respect to a tensor basis having minimal
support. The orthogonal basis establishes the logarithmic consistent
projection for the orthogonal windowing of pairwise comparisons matrices.
It is compared with the windowing of the Saaty and SVD types. These
comparisons resulted in new composite formulae for logarithmic, Saaty,
and SVD projections. The theoretical considerations presented in the
paper are accompanied by numerous examples.}

\keywords{pairwise comparisons, consistent subspaces, tensor bases, basis orthogonalization,
matrix approximation, consistent projections}

\maketitle

\section{Introduction}

Multiplicative pairwise comparison matrices constitute the fundamental
mechanism for representing preferences in the Analytic Hierarchy Process,
developed by Thomas L. Saaty \citep{Saaty1977asmf}. Their purpose
is to formally capture a decision-maker’s subjective evaluations as
quantitative relationships between pairs of compared elements, such
as criteria or decision alternatives. The pairwise comparison matrix
$M$ is square and positive, and each of its elements $m_{ij}$ expresses
the relative importance of element i with respect to element j. The
interpretation of these values is quotient-based: the number $m_{ij}$
indicates how many times one element is more important than the other.
Consequently, the matrix is multiplicative in nature, since its elements
can be interpreted as quotients of unknown weights $w_{i}$, i.e.,
$m_{ij}=x_{i}/x_{j}$. An important property of these matrices is
the reciprocity of their elements, meaning that $m_{ij}=1/m_{ji}$.
Thanks to this, the information contained in the matrix is directionally
consistent: if one element is rated as more important than another,
the inverse relationship is unambiguously defined. The diagonal of
the matrix consists of ones, because each element is compared with
itself. In the case of ideal consistency of ratings, the condition
of multiplicative transitivity also holds, according to which the
relationships between three elements should satisfy the equation $m_{ij}\cdot m_{jk}=m_{ik}$.
In practice, however, expert ratings rarely satisfy this condition
exactly, leading to some inconsistency. 

We can naturally convert multiplicative matrices to additive matrices
by taking logarithms. After this transformation, the elements of matrix
$M=[m_{ij}]$ become $\ln m_{ij}$, which is close to the difference
of the logarithms of the weight vector elements, i.e., $m_{ij}\approx\ln x_{i}-\ln x_{j}$.
Thanks to this transformation, it becomes possible to use linear algebra
to describe and solve the problem of finding a priority vector from
the comparative data encoded in the matrix. This problem can be solved
in many ways, ranging from optimization methods \citep{Lootsma1985peon,Bozoki2010ooco}
and iterative inconsistency-reduction algorithms \citep{Koczkodaj2015fcod}
to algebraic transformations, as presented in this paper. 

Consistency plays a central role in the analytic hierarchy process
(AHP) \citep{Kulakowski2020utahp}. Within the framework of AHP, one
approximates a given matrix $M$ in the group ${\mathcal{M}}_{n}=({\mathcal{M}}_{n},\cdot)$
of all $\text{PC}$ matrices endowed with the Birkhoff dot product
by an appropriate consistent matrix:

\noindent
\[
x\otimes x^{-1}\overset{\textit{df}}{=}{\left[{x_{i}}/{x_{j}}\right]}^{n}_{i,j=1},\ \ x={\left(x_{1},\dots,x_{n}\right)}^{T}\in{\mathbb{R}}^{n}_{+},
\]

\noindent from the subgroup $\mathcal{C}_{n}\subset\mathcal{M}_{n}$
of all consistent matrices. For this purpose, several approximation
methods can be found in the literature. The earliest and most important
methods include: the eigenvector method according to Saaty \citep{Saaty1977asmf,Saaty1980tahp},
the least-squares method\textit{ }presented by Chu et al. \citep{Chu1979acot},
the logarithmic least-squares method introduced by Lootsma \citep[p. 96]{Lootsma1985peon}
and studied by van Laarhoven and Pedrycz \citep{VanLaarhoven1983afeo},
Crawford and Williams \citep{Crawford1985anot}, as well as Koczkodaj
and Orłowski \citep{Koczkodaj1997aobf}, and finally the singular
value decomposition method presented by Gass and Rapcsak \citep{Gass2004svdi}.
We note that some interesting comparative analyses of the above approximation
methods were given by Bozoki and Rapcsak \citep{Bozoki2008osak},
Gass and Rapcsak \citep{Gass2004svdi}, and Golany and Kress \citep{Golany1993ameo}.
Since then, several new papers on consistent approximation of $\text{PC}$
matrices have emerged within AHP. For example, we mention a few recent
papers presented by Bozoki \citep{Bozoki2014iwfp}, Koczkodaj et al.
\citep{Koczkodaj2020oopo}, Liu et al. \citep{Liu2019atma}, Smarzewski
and Rutka \citep{Smarzewski2020cpai}, Smarzewski and Kozera \citep{Kozera2022ccai},
and Magnot et al. \citep{Magnot2023agmf}. We note that the third
paper contains an excellent review of existing methods for measuring
inconsistency in AHP.

\noindent It should be noted that Saaty's eigenvector method is widely
used in applications of $AHP$ in many areas, including psychology,
product management, strategic planning, finance and banking, market
research, and others \citep{Golany1993ameo,Saaty1980tahp,Sipahi2010tahp,Ho2018tsot,Hyde2012maot}.
This is still the case in spite of criticism of the eigenvector approach
to pairwise comparisons found in the papers of Blanquero et al. \citep{Blanquero2006iewf},
Gass and Rapcsak \citep{Gass2004svdi}, Bozoki \citep{Bozoki2014iwfp},
Bozoki and Rapcsak \citep{Bozoki2008osak}, and the references therein.
Janicki with Koczkodaj and later with Zhai \citep{Janicki2011ropc,Janicki1996awoa}
propose a shift from the classical, numerical pairwise comparison
matrix to a non-numerical representation based on partial order relations,
in which qualitative preference relations replace numerical values.
This transformation involves mapping the matrix's value intervals
to discrete relations (e.g., \textquotedbl slightly better,\textquotedbl{}
\textquotedbl significantly better\textquotedbl ), yielding a relational
rather than a metric structure. An important result is that for consistent
matrices, one can construct a corresponding consistent system of partial
order relations, demonstrating the consistency of both approaches
at the structural level. At the same time, the non-numerical approach
eliminates the need for precise preference scaling, making it more
appropriate for qualitative evaluations. In \citep{Janicki2018fcwa},
Janicki expands on his earlier relational approach by integrating
qualitative and quantitative models into a unified procedure for determining
weights and rankings. The second work proposes a practical algorithm
that combines multiplicative, additive, and relational comparisons,
supplemented by mechanisms to reduce inconsistencies and iteratively
reconcile ratings. Chu contributed the paper \citep{Chu1998otoc},
which focuses on the problem of finding the optimal approximation
of an inconsistent pairwise comparison matrix by a consistent matrix
in the least-squares sense. The author shows that, after transitioning
to an additive (logarithmic) representation, the space of consistent
matrices has a linear structure, thereby significantly simplifying
the mathematical analysis and computations. Chu derives the optimality
conditions and proposes effective numerical methods (including quasi-Newton
methods) for solving the approximation problem. Chu's study serves
as an important link between classical AHP approaches, optimization
tools, and linear algebra. Fedrizzi, Brunelli and Caprila \citep{Fedrizzi2020tlao}
presents a framework for pairwise comparison matrices within the context
of linear algebra, using an additive representation of preferences.
The authors analyze three related vector spaces: the space of all
comparison matrices, its subspace of consistent matrices, and its
orthogonal complement, referred to as “totally inconsistent” matrices.
In their paper, the authors also present the decomposition of any
matrix into a consistent and inconsistent part, along with the interpretation
of the coordinates of this decomposition. At the end of the paper,
the authors analyze Csató transformation \citep{Csato2019acot} in
terms of coordinates and orthogonal projections. 

In our work, as in \citep{Chu1998otoc,Fedrizzi2020tlao}, we start
from the same observation that the logarithmic transformation allows
us to apply linear algebra to the analysis of pairwise comparison
matrices. In our discussion, however, we go a step further and construct
its first explicit orthogonal basis, derive closed-form expressions
for the orthogonal projection, and show how this construction leads
to new formulas for the logarithmic projection, as well as to a comparison
of it with the Saaty projection and the SVD projection.

\noindent In our study, we consider the classical Gram-Schmidt orthogonalization
problem \citep{Deutsch2001baii} for the linear $\left(n-1\right)$
-- dimensional space $\mathcal{A}_{n}=\ln\mathcal{C}_{n}$ of all
additively consistent matrices in $\mathbb{R}^{n\times n}$:

\noindent
\[
x\ominus x\overset{\textit{df}}{=}{\left[x_{i}{-x}_{j}\right]}^{n}_{i,j=1},\ \ x={\left(x_{1},\dots,x_{n}\right)}^{T}\in{\mathbb{R}}^{n},
\]

\noindent endowed with the weighted Frobenius inner product, 
\[
\left\langle A,B\right\rangle =\sum^{n}_{i,j=1}\varrho_{i}\varrho_{j}(x_{i}-x_{j})(y_{i}-y_{j}),
\]
where $\varrho=\left(\varrho_{1},\dots,\varrho_{n}\right)^{T}$ is
a given vector of positive weights. It can be formulated as follows. 
\begin{problem}
\label{prob:problem-1-1}Given a basis $\left\{ A_{1},\dots,A_{n-1}\right\} $
of ${\mathcal{A}}_{n}$, find orthogonal basis $\left\{ B_{1},\dots,B_{n-1}\right\} $
of ${\mathcal{A}}_{n}$ such that 
\[
\textrm{span}\left\{ A_{1},\dots,A_{k}\right\} =\textrm{span}\left\{ B_{1},\dots,B_{k}\right\} ,
\]
for each $k=1,\dots,n-1.$ Moreover, find useful formulae for numerical
calculation of the orthogonal basis $\left\{ B_{k}\right\} $.
\end{problem}
\noindent It is well-known that (the first part) of Problem \ref{prob:problem-1-1}
has a unique solution $\left\{ B_{k}\right\} $, up to nonzero multiplicative
constants. The solution can be computed by recurrence formulae due
to Gram-Schmidt, which almost never are efficient numerically. In
order to illustrate the requirement of efficiency in Problem \ref{prob:problem-1-1},
we mention that Gram-Schmidt recurrence relations for monomial basis
$\left\{ 1,x,\dots,x^{n}\right\} $ are considerably improved by more
efficient numerically three-term recurrence relations for orthogonal
polynomial bases \citep{Deutsch2001baii}.

\section{Main results}

\noindent In this paper, we study Problem \ref{prob:problem-1-1}
for the basis $\left\{ A_{1},\dots,A_{n-1}\right\} $ of $\mathcal{A}_{n}$.
In the unweighted case, this is a long standing problem due to Koczkodaj
and Orłowski, who investigated it in the paper \citep{Koczkodaj1997aobf}
on the logarithmic least-squares method in pairwise comparisons.

\noindent Recently, this problem has been considered in \citep{Koczkodaj2020oopo,Smarzewski2020cpai,Benitez2024ceof},
but efficient explicit expressions for the orthogonal basis $\left\{ B_{1},\dots,B_{n-1}\right\} $
have not been found until now. The present successful attempt has
been possible due to a tensor description of the explicit weighted
formulae \citep{Koczkodaj2020oopo,Smarzewski2020cpai} for best additively
consistent approximations of skew-symmetric matrices.

\noindent More precisely, in Theorem \ref{th:every-additively-consistent-matrix-4-2},
we focus on main properties of the consistent tensor basis $\left\{ A_{1},\dots,A_{n-1}\right\} $
of ${\mathcal{A}}_{n}$ having minimal support:
\[
A_{k}=e_{k}\ominus e_{k}\ \ \ \left(k=1,\dots,n-1\right),
\]

\noindent while in Corollary \ref{cor:For-any-skew-symmetric-5-4}
we solve Problem \ref{prob:problem-1-1} by giving the following most
efficient orthogonal consistent tensor basis $\left\{ B_{1},\dots,B_{n-1}\right\} $
of $\mathcal{A}_{n}$:
\[
B_{k}=x^{k}\ominus x^{k}\,\,\,\left(k=1,\dots,n-1\right),
\]

\noindent where $x^{k}={\left(x^{k}_{1},\dots,x^{k}_{n}\right)}^{T}$
are vectors in ${\mathbb{R}}^{n}$ defined by

\noindent
\[
x^{k}_{1}=\dots=x^{k}_{k-1}=\frac{1}{n-k+1},\ \ x^{k}_{k}=1,\ \ x^{k}_{k+1}=\dots=x^{k}_{n}=0.
\]

\noindent It should be noted that Koczkodaj and Orłowski \citep[p. 43]{Koczkodaj1997aobf}
discovered experimentally the following basis ${\hat{A}}_{k}=[{\hat{a}}^{k}_{ij}]$
of ${\mathcal{A}}_{n}$ defined by

\noindent
\[
\hat{a}^{k}_{ij}=\left\{ \begin{array}{c}
1,\ \ \mathrm{for}\ 1\le i\le k<j\le n,\\
-1,\ \mathrm{\ for}\ 1\le j\le k<i\le n,\\
\ \ 0,\ \ \ \ \mathrm{\textrm{otherwise}},\ \ \ \ \ \ \ \ \ \ \ \ \ \ \ \ \ \ \ \ \ 
\end{array}\right.
\]

\noindent which does not have the minimal supports:

\noindent
\[
\text{supp}\left({\hat{A}}_{k}\right)=\left\{ \left(i,j\right):{\hat{a}}^{k}_{ij}\neq0\right\} ,A\in{\mathbb{R}}^{n\times n}.
\]

\noindent In fact, each basis matrix ${\hat{A}}_{k}$ has $2k\left(n-k\right)$
entries ${\hat{a}}^{k}_{ij}\neq0$, while $A_{k}$ has only $2\left(n-1\right)$
entries different from zero. For this basis, the authors of \citep{Koczkodaj1997aobf}
claimed that a solution of Problem \ref{prob:problem-1-1} consists
of orthogonal matrices ${\hat{B}}_{k}$ satisfying recurrence relations
\[
{\hat{B}}_{k}={\hat{A}}_{k}-\frac{n-k}{n-k+1}{\hat{A}}_{k-1}\ \left(k=1,\dots,n-1\right),
\]
where ${\hat{A}}_{0}$ is a matrix with all zero elements. However,
they were only able to prove these formulae for an easy case $n=4.$
We conjecture that the recurrence relations are incorrect for higher
values of $n$.
\begin{remark}
\noindent The tensor basis $\left\{ A_{1},\dots,A_{n-1}\right\} $
of ${\mathcal{A}}_{n}$ was already proposed by Chu \citep[p. 158]{Chu1998otoc}.
Fedrizzi, Brunelli, and Caprila also noticed this and incorporated
it into their work \citep[p. 192]{Fedrizzi2020tlao}.
\end{remark}
\noindent In this paper the authors characterized the orthogonal direct
sum decomposition $S_{n}={\mathcal{A}}_{n}\bigoplus{\mathcal{A}}^{\bot}_{n}$
of $\frac{1}{2}n\left(n-1\right)$ -- dimensional space $S_{n}$
of a all skew-symmetric matrices in $\mathbb{R}^{n\times n}$, endowed
with the weighted Frobenius inner product. For this purpose, the authors
introduced two suitable bases for spaces ${\mathcal{A}}_{n}$ and
$\mathcal{A}^{\bot}_{n}$. Although both these bases are orthogonal
with respect to some inner products \citep{Fedrizzi2020tlao}, it
does not help to solve Problem \ref{prob:problem-1-1} neither for
basis $\left\{ \hat{A}_{1},\dots,\hat{A}_{n-1}\right\} $ nor for
simpler basis $\left\{ A_{1},\dots,A_{n-1}\right\} .$ It has been
already noticed by Benitez, Koczkodaj and Kowalczyk \citep{Benitez2024ceof},
who reconsidered characterization of orthogonal direct sum $S_{n}={\mathcal{A}}_{n}\bigoplus{\mathcal{A}}^{\bot}_{n}$
, but were not able to simplify the original recurrence relations
of Gram-Schmidt\textbf{. }

\noindent It should be noticed that our construction of an orthogonal
basis has shed new light not only on the orthogonal projection \citep{Koczkodaj2020oopo}
of PC matrices, but also on the Saaty \citep{Saaty1977asmf} and SVD
projections. The third projection has been inspired by an interesting
paper of Gass and Rapcsak \citep{Gass2004svdi}. These three kinds
of PC projections are studied both from theoretical and computational
points of view in Sections \ref{sec:Consistent-approximation-of},
\ref{sec:Orthogonal-windowing-of}, \ref{sec:Saaty-windowing-of}
and \ref{sec:SVD-windowing-of}. It is clear that each of them directly
yields a pairwise comparisons method. In the following sections, it
is called windowing. To be more precise, let ${\mathcal{R}}_{\mathcal{C}}$,\textit{
}${\mathcal{P}}_{\mathcal{C}}{\mathrm{=exp}({\mathcal{P}}_{\mathcal{A}}({\ln M\ }))\ }$
and ${\mathcal{T}}_{\mathcal{C}}$ denote respectively the eigenvector,
logarithmic and SVD projections of ${\mathcal{M}}_{n}$ onto ${\mathcal{C}}_{n}.$
Then we show (Theorem \ref{thm:For-any-matrix-6-2}) the following
windowing formula:

\noindent
\[
{\mathcal{P}}_{\mathcal{C}}M=\prod^{n-1}_{k=1}{{\ [M,B_{k}]}^{\ \frac{n-k+1}{2n(n-k)}B_{k}}}\ \ \ 
\]

\noindent for an orthogonal projection of $PC$ matrices $M$ in ${\mathcal{M}}_{n}$.
Moreover, we prove that the SVD projection ${\mathcal{T}}_{\mathcal{C}}$
is a composite projection of the form

\noindent
\[
{\mathcal{T}}_{\mathcal{C}}M={\mathcal{R}}_{\mathcal{C}}\left({\mathcal{T}}_{\mathcal{G}}M\right)={\mathcal{P}}_{\mathcal{C}}\left({\mathcal{T}}_{\mathcal{G}}\mathrm{\ }\mathrm{M}\right),\ M\in{\mathcal{M}}_{n},
\]

\noindent where ${\mathcal{T}}_{\mathcal{G}}$ is the tensor projections
of ${\mathcal{M}}_{n}$ onto the group ${\mathcal{G}}_{n}=\left({\mathcal{G}}_{n},\cdot\right)$
of tensor products $x\otimes y=xy^{T}$ of positive column vectors
$x,y\in{\mathbb{R}}^{n}.$ Hence, it is of interest that the SVD projection
is a unique pairwise comparisons method which is both norm and eigenvector
based. This fact follows from the following independently interesting
identity:

\noindent
\[
{\mathcal{R}}_{\mathcal{C}}\left(x\otimes y\right)=\mathcal{T}_{\mathcal{C}}\left(x\otimes y\right),
\]

\noindent which is proved in Lemma \ref{lem:lem-8-6}. In Sections
\ref{sec:Saaty-windowing-of} and \ref{sec:SVD-windowing-of} there
are other similar results, e.g., for normal, stochastic and doubly
stochastic matrices. They all support the conclusion of Golany and
Kres \citep{Golany1993ameo} that none of the windowing methods is
better in all cases. Thus, it is an additional argument against the
claim of Saaty and Vargas \citep{Saaty1984coel} about the absolute
superiority of the eigenvector method.

\noindent Further theoretical studies along these lines are needed,
as well as numerical comparisons of orthogonal and Saaty windowing
with the SVD windowing of PC matrices in real-world applications.
Of course, new numerical experiments may take into account many existing
comparisons of the logarithmic least squares method with the eigenvector
method presented in hundreds of papers and books, for example, in
\citep{Kulakowski2022otsb,Bozoki2008osak,Golany1993ameo,Golden1989tahp,Liu2020dpfp,Saaty1980tahp,Saaty1984coel,Sipahi2010tahp}.
These references contain several examples of AHP's real-world applications.

At a fundamental level, projections onto the subspace of consistent
matrices can be interpreted as weighting methods. This follows from
the well-known one-to-one correspondence between consistent pairwise
comparison matrices and priority vectors: each consistent matrix uniquely
determines a weight vector (each column of the consistent PC matrix,
after appropriate rescaling, is equivalent to a priority vector),
and conversely, each weight vector $w=[w_{1},\ldots,w_{n}]^{T}$ uniquely
induces a consistent matrix $A=[a_{ij}]$, where $a_{ij}=w_{i}/w_{j}$.
Therefore, any procedure that projects an inconsistent matrix onto
the set of consistent matrices can naturally be understood as a method
for deriving weights. From this perspective, projection-based approaches
provide a unifying framework for understanding and computing priority
vectors. In light of the above, a number of the results presented
here, despite their formal nature, have practical applications. In
particular, Theorem \ref{the:theorem-3-2} provides an explicit solution
to the problem of logarithmic approximation in the Frobenius norm
via a priority vector. Theorem \ref{thm:thm-5-3}, together with Corollary
\ref{cor:For-any-skew-symmetric-5-4}, constructs an explicit basis
and a closed-form formula for the orthogonal projection $\mathcal{P}_{\mathcal{A}}$,
which can be computed via a finite sum of scalar products -- this
is a direct algorithm of polynomial complexity. Similarly, Theorem
\ref{thm:For-any-matrix-6-2} provides a closed-form expression for
the projection $\mathcal{P}_{\mathcal{C}}$ in the multiplicative
model as a product of factors dependent on the input matrix, which
eliminates the need for iterative methods and allows for a stable
numerical implementation. Finally, Theorem \ref{thm:thm-8-2} shows
that the SVD projection can be understood as a composition of two
simpler projections, which paves the way for a modular implementation
of the solution.

\section{Tensor notation of consistent matrices\label{sec:tensor-notation-of}}

\noindent In this paper, all vectors are understood as column vectors.
Let

\noindent
\[
x=\left(x_{1},\dots,x_{n}\right)^{T}\ \ \text{and}\ \ y=\left(y_{1},\dots,y_{n}\right)^{T}
\]

\noindent be two vectors in ${\mathbb{R}}^{n},$ where $T$ means
transpose. Then the statement $x<y$ means that $x_{i}<y_{i}\ \left(i=1,\dots,n\right).$
Moreover, their inner product $\left\langle \cdot,\cdot\right\rangle :{\mathbb{R}}^{n}{\mathrm{\times}\mathbb{R}}^{n}\longrightarrow\mathbb{R}$,
tensor product ${\otimes\mathrm{:}\mathbb{R}}^{n}\times{\mathbb{R}}^{n}\longrightarrow{\mathbb{R}}^{n\times n}$
and tensor sum $\oplus\mathrm{:}\mathbb{R}^{n}\times\mathbb{R}^{n}\longrightarrow\mathbb{R}^{n\times n}$
are defined as follows:

\noindent
\[
\left\langle x,y\right\rangle \overset{\textit{df}}{=}x^{T}y=\ x_{1}y_{1}+\dots+x_{n}y_{n},
\]

\begin{equation}
x\otimes y\overset{\textit{df}}{=}xy^{T}=\left[\begin{array}{ccc}
x_{1}y_{1} & \cdots & x_{1}y_{n}\\
\vdots & \ddots & \vdots\\
x_{n}y_{1} & \cdots & x_{n}y_{n}
\end{array}\right],\label{eq:tensor-product-def}
\end{equation}

\[
x\oplus y\overset{\textit{df}}{=}\left[\begin{array}{ccc}
x_{1}+y_{1} & \cdots & x_{1}+y_{n}\\
\vdots & \ddots & \vdots\\
x_{n}+y_{1} & \cdots & x_{n}+y_{n}
\end{array}\right].
\]

\noindent Note that the tensor product and sum represent endomorphisms
acting on $\mathbb{R}^{n}$, i.e., 

\noindent
\[
\left(x\otimes y\right)\left(z\right)=xy^{T}z=x\left\langle y,z\right\rangle =\left\langle z,y\right\rangle x,
\]

\[
\left(x\oplus y\right)\left(z\right)=\left\langle z,e\right\rangle x+\left\langle z,y\right\rangle e,
\]

\noindent where $z\in{\mathbb{R}}^{n}$ and $e={\left(1,\dots,1\right)}^{T}\in{\mathbb{R}}^{n}.$

\noindent A positive matrix $M=\left[m_{ij}\right]$ in ${\mathbb{R}}^{n\times n}$
is said to be a PC matrix if it is reciprocal, i.e., if we have

\noindent
\begin{equation}
m_{ij}=1/m_{ji}\ \ \ \ (i,j=1,\dots,n).\label{eq:reciprocal-matrix-def}
\end{equation}

\noindent If a PC matrix $C=\left[c_{ij}\right]\in{\mathbb{R}}^{n\times n}$
satisfies the transitivity condition

\noindent
\begin{equation}
c_{ij}=c_{ik}c_{kj}\ \ \ \left(1\le i<k<j\le n\right),\label{eq:consistency-cond-def}
\end{equation}

\noindent then it is called consistent. Denote the set of all consistent
matrices by ${\mathcal{C}}_{n}$ and note that ${\mathcal{C}}_{n}=({\mathcal{C}}_{n},\cdot)$
is a subgroup of the abelian group ${\mathcal{M}}_{n}=({\mathcal{M}}_{n},\cdot)$
of all $PC$ matrices, where the dot is the Hadamard product (entry-wise
product) of two matrices $M=\left[m_{ij}\right]$ and $N=\left[n_{ij}\right]$
defined by $M\cdot N=\left[m_{ij}n_{ij}\right].$ Similarly, we define
the Hadamard quotient (entry-wise quotient) ${x}/{y={({x_{1}}/{y_{1},\dots,}{x_{n}}/{y_{n}})}^{T}}$
of two vectors $x=\left(x_{1},\dots,x_{n}\right)^{T}$ and $y=\left(x_{1},\dots,x_{n}\right)^{T}>0$.

\noindent It is well-known that a matrix $C=\left[c_{ij}\right]$
is consistent if and only if there exists a positive vector $x\,$
such that $C=[{x_{i}}/{x_{j}]}.$ The vector $x\,$ is uniquely determined
by the consistent matrix $C,$ up to a positive constant factor, and
is known as the priority vector of $C$. Thus a consistent matrix
$C\in{\mathcal{C}}_{n}$ can be rewritten as the following tensor
product or tensor quotient:

\noindent
\begin{equation}
C=x\otimes x^{-1}=x\oslash x=[x_{i}/x_{j}],\label{eq:tensor-quotient}
\end{equation}

\noindent where $x={\left(x_{1},\dots,x_{n}\right)}^{T}$ is the priority
vector of $C$ and $x^{-1}=\left(1/x_{1},\dots,1/x_{n}\right)^{T}$
is the group inverse of $x\,$ in ${\mathcal{M}}_{n}=\left({\mathcal{M}}_{n},\cdot\right).\ $
Note that if we choose $x_{n}=\alpha>0$, then the remaining components
of the priority vector $x$ are equal to $x_{i}=\alpha c_{in}.$

\noindent The multiplicative group ${\mathcal{M}}_{n}=({\mathcal{M}}_{n},\cdot)$
is isomorphic to the additive group ${\mathcal{S}}_{n}=({\mathcal{S}}_{n},+)$
of all skew-symmetric matrices $S=[s_{ij}]$ in ${\mathbb{R}}^{n\times n}$
endowed with the usual entry-wise addition of matrices. Since (\ref{eq:reciprocal-matrix-def})
is equivalent to the equality $\ln m_{ij}=-\ln m_{ji}$, it follows
that the matrix logarithmic function defined by $S=\ln M=\left[\ln m_{ij}\right]$
is a $\mathrm{group\ isomorphism}$ of ${\mathcal{M}}_{n}$ onto ${\mathcal{S}}_{n}$
and that the matrix exponential function $M=\exp S=\left[\exp s_{ij}\right]$
is its inverse. The group ${\mathcal{S}}_{n}=({\mathcal{S}}_{n},+)$
contains a subgroup $\mathcal{A}_{n}=(\mathcal{A}_{n},+)$ of all
additively consistent matrices $A=[a_{ij}]$ defined by the condition

\noindent
\begin{equation}
a_{ij}=a_{ik}+a_{kj}\ \ \ (1\le i<k<j\le n).\label{eq:additive-consistency}
\end{equation}

\noindent It is clear that a skew-symmetric matrix $A=\left[a_{ij}\right]$
is additively consistent if and only if there exists a column vector
$y={\left(y_{1},\dots,y_{n}\right)}^{T}$such that the matrix $A$
can be represented as the following tensor sum or tensor difference:\textit{ }

\noindent
\begin{equation}
A=y\oplus\left(-y\right)=y\circleddash y=[y_{i}-y_{j}],\label{eq:add-consistent-matrix-as-tensor-difference}
\end{equation}

\noindent where $-y={\left(-y_{1},\dots,-y_{n}\right)}^{T}$ is the
group inverse of $y$ in ${\mathcal{S}}_{n}=\left({\mathcal{S}}_{n},+\right).$
The vector $y$ is uniquely determined by $A,$ up to an additive
constant, and is known as the additive priority vector\textit{ }of
$A.\ $ If we choose $y_{n}=\beta$, then the remaining components
of $y$ are $y_{i}=\beta+a_{in}.$

\noindent We note that if a matrix $A=y\circleddash y=\left[y_{i}-y_{j}\right]$
is additively consistent, then the matrix $C=\exp A=[\exp y_{i}/\exp y_{j}]$
is consistent. Conversely, if a matrix $C=x\oslash x=[{x_{i}}/{x_{j}}]$
is consistent, then the matrix $A=\ln C=\left[\ln x_{i}-\ln x_{j}\right]$
is additively consistent. Thus subgroups ${\mathcal{A}}_{n}=({\mathcal{A}}_{n},+)$
and ${\mathcal{C}}_{n}=({\mathcal{C}}_{n},\cdot)$ are isomorphic,
under the entry-wise logarithmic or exponential functions. We denote
them by using the symbols $\mathcal{A}_{n}=\ln\mathcal{C}_{n}$ and
$\mathcal{C}_{n}=\exp\mathcal{A}_{n}$.

\section{Consistent approximation of PC matrices\label{sec:Consistent-approximation-of}}

\noindent Within the framework of AHP, one approximates a given PC
matrix $M$ from the multiplicative group ${\mathcal{M}}_{n}=({\mathcal{M}}_{n},\cdot)$
by an appropriate consistent matrix $C=x\oslash x$ from the subgroup
$\mathcal{C}_{n}=\left(\mathcal{C}_{n},\cdot\right)$. For this purpose,
we suppose that $\left\Vert \cdot\right\Vert $ is a matrix norm in
${\mathbb{R}}^{n\times n}$ and define the distance from $M\in{\mathcal{M}}_{n}$
to ${\mathcal{C}}_{n}$ by

\noindent
\begin{equation}
\text{dist}\left(M,\mathcal{C}_{n}\right)=\mathop{\mathrm{inf}}_{x>0}\left\Vert M-x\oslash x\right\Vert .\label{eq:matrix-distance}
\end{equation}

\begin{problem}
\noindent\label{prob:Given-a-PC-3-1}Given a $PC$ matrix $M$ in
the group ${\mathcal{M}}_{n}=\left({\mathcal{M}}_{n},\cdot\right)$,
find a matrix $\hat{x}\oslash\hat{x}$ in the subgroup ${\mathcal{C}}_{n}=\left({\mathcal{C}}_{n},\cdot\right)$
such that

\noindent
\begin{equation}
\left\Vert M-\hat{x}\oslash\hat{x}\right\Vert =\text{dist}\left(M,\mathcal{C}_{n}\right).\label{eq:eq-3-2}
\end{equation}
\end{problem}
\noindent The matrix $C=\hat{x}\ \oslash\hat{x}$ is called a best
consistent approximation, or a nearest consistent matrix from ${\mathcal{C}}_{\mathrm{n}}$
to $M.$ Although the set ${\mathcal{C}}_{n}$ is neither bounded
nor closed in ${\mathbb{R}}^{n\times n},$ the consistent approximation
problem is always solvable. In fact, we have
\begin{theorem}
\noindent\label{th:on-best-consistent-approx}If $M\in{\mathcal{M}}_{n}$,
then there exists a best consistent approximation $C=\hat{x}\oslash\hat{x}$
from ${\mathcal{C}}_{\mathrm{n}}$ to $M.$
\end{theorem}
\begin{proof}
By (\ref{eq:eq-3-2}) we have 
\[
\left\Vert \hat{x}\oslash\hat{x}\right\Vert \le\left\Vert M-\hat{x}\oslash\hat{x}\right\Vert +\left\Vert M\right\Vert \le\left\Vert M-e\oslash e\right\Vert +\left\Vert M\right\Vert .
\]

\noindent Hence the vector $\hat{x}$ should belong to a nonempty
bounded closed subset $G$ of all positive vectors $x\in\mathbb{R}^{n}$,
which satisfy the inequality 
\[
\left\Vert x\oslash x\right\Vert \le\left\Vert M-e\ \oslash e\right\Vert +\left\Vert M\right\Vert .
\]

\noindent By the Bolzano$-$Weierstrass theorem \citep{Rusnock2005bauc}
the set $G$ is compact. Thus, the extreme value theorem \citep[p. 89]{Rudin1976poma}
should be applied to conclude that the continuous function 
\[
\varrho\left(x\right)=\left\Vert M-x\oslash x\right\Vert \ \ \ \ (x\in G)
\]

\noindent attains its minimum at some $\hat{x}$.
\end{proof}
\noindent The consistent approximation problem is non-convex, and
so it may have more than one solution \citep{Smarzewski2020cpai}.
Thus, it may be ill--posed numerically. Such a situation does not
occur in the case of the eigenvector method according to Saaty \citep{Saaty1977asmf,Saaty1980tahp}.
More precisely, we have
\begin{remark}
\noindent Let $A=[a_{ij}]$ be a positive matrix in ${\mathbb{R}}^{n\times n},$
and let a positive continuous functional $r_{A}$ be defined on $R^{n}_{+}$
by the formula
\[
r_{A}(x)=\mathop{\max}_{1\le i\le n}\frac{(Ax)_{i}}{x_{i}},\ \ (Ax)_{i}=\sum^{n}_{j=1}a_{ij}x_{j}.
\]

\noindent As in Theorem (\ref{th:on-best-consistent-approx}), we
conclude that this non-subadditive functional $r_{A}$ attains a minimum
at some $\hat{x}$: 
\[
r_{A}(\hat{x})=\mathop{\min}_{x>0}r_{A}(x).
\]

\noindent According to Saaty \citep{Saaty1977asmf,Saaty1980tahp},
a consistent approximation from ${\mathcal{C}}_{\mathrm{n}}$ to $M$
is defined again by the formula $C=\hat{x}\oslash\widehat{x.}$ Saaty's
method is based on the observation (indicated in Perron's theorem)
according to which \citep{Perron1907ctdm,Gantmaher2000ttom} the positive
vector $\hat{x}$ is the unique, up to the positive constant, eigenvector
of the matrix $A$, corresponding to the unique largest eigenvalue
$\lambda_{\max}=r_{A}\left(\hat{x}\right)$ of $A.$
\end{remark}
\noindent It is relatively hard to solve Problem \ref{prob:Given-a-PC-3-1}
even in the simplest case of the Frobenius norm of matrices $X=[x_{ij}]$
in ${\mathbb{R}}^{n\times n}:$ 
\[
\left\Vert X\right\Vert _{F}=\left(\sum^{n}_{i,j=1}x^{2}_{ij}\right)^{1/2}.
\]

\noindent This was already observed by Saaty in \citep{Saaty1980tahp}
and later discussed by Saaty and Vargas in \citep{Saaty1984coel}.
Therefore, following Lootsma \citep{Lootsma1981peon} and van Laarhoven
and Pedrycz \citep{VanLaarhoven1983afeo}, it is reduced to a linear
approximation problem which is well-posed and much easier to solve,
not only in the particular case of the Frobenius norm.
\begin{problem}
\noindent Given a positive matrix $M$ in ${\mathbb{R}}^{n\times n},$
find a matrix $\exp(\hat{y}\circleddash\hat{y})$ in the consistent
group ${\mathcal{C}}_{n}=\left({\mathcal{C}}_{n},\cdot\right)$ such
that the matrix $\hat{y}\ominus\hat{y}$ is nearest to the matrix
$S=\ln M$ from the additively consistent subgroup $\mathcal{A}_{n}$$=\left({\mathcal{A}}_{n},+\right),$
i.e.,

\noindent
\begin{equation}
\left\Vert S-\hat{y}\circleddash\hat{y}\right\Vert =\text{dist}\left(S,\mathcal{A}_{n}\right),\label{eq:eg-3-3}
\end{equation}

\noindent where

\noindent
\[
\text{dist}\left(S,\mathcal{A}_{n}\right)=\mathop{\mathrm{inf}}_{y\in\mathbb{R}^{n}}\left\Vert S-y\circleddash y\right\Vert .
\]

\noindent If we take $w=\mathrm{exp\ (}\hat{y}\ ),$ then the matrix
$C=w\oslash w$ is called the logarithmic consistent approximation
from $\mathcal{C}_{\mathrm{n}}$ to $\mathrm{M.}$
\end{problem}
\noindent Since ${\mathcal{A}}_{n}$ is a subspace of finite dimensional
space ${\mathcal{S}}_{n}$ of all skew--symmetric matrices in $\mathbb{R}^{n\times n}$,
it follows that the additive approximation problem introduced above
has at least one solution $\hat{y}\circleddash\hat{y}$ for every
$S=\ln M.$ This solution is unique, whenever the matrix norm $\left\Vert \cdot\right\Vert $
is strictly convex. In particular, it is true in the case of the weighted
Frobenius norm $\left\Vert X\right\Vert _{F(\varrho)}=\sqrt{\left\langle X,X\right\rangle }$
induced by the weighted Frobenius inner product 
\begin{equation}
\left\langle X,Y\right\rangle =\sum^{n}_{i,j=1}\varrho_{i}\varrho_{j}x_{ij}y_{ij},\label{eq:weighted-frobenius-inner-product}
\end{equation}

\noindent where $\varrho={\left({\varrho}_{1},\dots,{\varrho}_{n}\right)}^{T}$
is a given vector of positive weights.

\noindent The weighted Frobenius norm ${\left\Vert M\right\Vert }_{F(\varrho)}$
coincides with the Frobenius norm ${\left\Vert S\cdot M\right\Vert }_{F(\varrho)}$
of the Hadamard product $S\cdot M$, where $S=\sqrt{\varrho\otimes\varrho}$
is a scaling matrix. This scaling method is two sided. Indeed, we
have 
\[
S\cdot M=DMD,
\]

\noindent where $D=\text{diag}\left(\sqrt{\varrho_{1}},\dots,\sqrt{\varrho_{n}}\right)$
is a diagonal matrix. The solution $w\oslash w$ of the logarithmic
consistent problem can be expressed in terms of generalized geometric
means. In fact, we have
\begin{theorem}
\noindent\citep{Koczkodaj2020oopo,Smarzewski2020cpai} \label{the:theorem-3-2}Let
${\left\Vert \cdot\right\Vert }_{F(\varrho)}$ be a weighted Frobenius
norm and let $M=[m_{ij}]$ be a positive matrix in ${\mathbb{R}}^{n\times n}.$
Then the solution $C=w\oslash w$ of the logarithmic consistent problem
is determined by a priority vector $w={(w_{1},\dots,w_{n})}^{T}$
such that 
\begin{equation}
w_{i}=\left(\prod^{n}_{k=1}m^{\varrho_{k}}_{ik}\right)^{1/\left|\varrho\right|}\,\,\text{for}\,\,i=1,\dots,n,\label{eq:eq-3-5}
\end{equation}

\noindent where $\left|\varrho\right|={\varrho}_{1}+\dots+{\varrho}_{n}.$ 
\end{theorem}
\noindent In real world applications of AHP, the entries of $M=[m_{ij}]$
are often treated as random variables associated to an unknown probability
distribution \citep{Bryson1995agpm}. Thus, the proper choice of weight
vector $\varrho={\left({\varrho}_{1},\dots,{\varrho}_{n}\right)}^{T}$
may be crucial while computing the priority vector $w$ for randomly
disturbed data. A first choice of ${\varrho}_{k}$ can be based on
the orthogonal polynomials theory. More precisely, we suggest taking
${\varrho}_{k}$ dependent on the weights of the inner products which
determine either discrete orthogonal classical polynomials of Chebyshev,
Hahn, Chartier, Mainer, and Kravchuk \citep{Rutka2023teep}, or continuous
orthogonal classical polynomials of Jacobi, Laguerre, Hermite, generalized
Bessel, Jacobi on $(0,+\infty)$, and pseudo$-$Jacobi kind \citep{Rutka2012csot}.
For this purpose, numerical algorithms are used for an extremely efficient
computation of the zeros of these polynomials; see Wang et al. \citep{Wang2014ebwf,Wang2011otcr}.
Another simpler choice of weights will be proposed after presentation
of the Saaty eigenvector principle in Section \ref{sec:SVD-windowing-of}
.

\section{Characterization of additive consistency\label{sec:Characterization-of-additive}}

\noindent The class ${\mathcal{A}}_{n}$ of additively consistent
matrices $A=[a_{ij}]$ is a subspace of $\frac{1}{2}n(n-1)$-dimensional
vector space ${\mathcal{S}}_{n}$ of all skew-symmetric matrices in
${\mathbb{R}}^{n\times n}.$ It is defined by $\frac{1}{6}n(n^{2}+5)$
linearly dependent constraints of the form: 
\begin{align}
a_{ij} & =-a_{ji}\ \ \ \left(1\le i\le j\le n\right),\nonumber \\
a_{ij} & =a_{ik}+a_{kj}\ \ \ (1\le i<k<j\le n).\label{eq:eq-4-1}
\end{align}

\noindent If we introduce the weighted Frobenius norm ${\left\Vert \cdot\right\Vert }_{F(\varrho)}$
induced by the weighted Frobenius inner product $\left\langle \cdot,\cdot\right\rangle _{F(\varrho)}$
in ${\mathbb{R}}^{n\times n},$ then $\mathcal{A}_{n}=\mathcal{P}_{\mathcal{A}}\left(\mathcal{S}_{n}\right)$
is a range of the linear orthogonal projection $\mathcal{P}_{\mathcal{A}}$
of $\mathcal{S}_{n}$ onto $\mathcal{A}_{n}$ defined by 
\begin{equation}
\left\Vert S-\mathcal{P}_{\mathcal{A}}S\right\Vert _{F(\varrho)}=\mathop{\min}_{A\in\mathcal{A}_{n}}\left\Vert S-A\right\Vert _{F(\varrho)}\boldsymbol{,\ \ }S\boldsymbol{\in}\mathcal{S}_{\boldsymbol{n}}.\label{eq:eq-4-2}
\end{equation}

\noindent This projection ${\mathcal{P}}_{\mathcal{A}}$ has been
determined by Koczkodaj et al. in a recent paper \citep{Koczkodaj2020oopo},
where the following lemma is proved.
\begin{lemma}
\noindent\label{lem:lem-4-1}If $S=[s_{ij}]\in\mathcal{S}_{n},$
then there exists a unique matrix $\mathcal{P}_{\mathcal{A}}S$ in
${\mathcal{A}}_{n}$ which satisfies (\ref{eq:eq-4-2}). The elements
of the matrix $\mathcal{P}_{\mathcal{A}}S=[y_{i}-y_{j}]$ are determined
by an additive priority vector $y=(y_{1},\dots,y_{n})^{T}$ such that
$y_{n}$ is an arbitrary real number and 
\begin{equation}
y_{i}=\frac{1}{\left|\varrho\right|}\sum^{n}_{j=1}\varrho_{j}s_{ij}\ \ \left(i=1,\dots,n-1\right),\label{eq:eq-4-3}
\end{equation}
\noindent\begin{flushleft}
where $\left|\varrho\right|=\varrho_{1}+\dots+\varrho_{n}$.
\par\end{flushleft}

\end{lemma}
\noindent This lemma yields a useful tensor representation of ${\mathcal{P}}_{\mathcal{A}}.$
Such a representation will now be derived, under the assumption that
the additive priority vector $y={(y_{1},\dots,y_{n})}^{T}$ is fixed,
by taking $y_{n}$ equal to the generalized arithmetic mean of the
n-th row of $S:$
\begin{equation}
y_{n}=\frac{1}{\left|\varrho\right|}\sum^{n}_{j=1}\varrho_{j}s_{nj}.\label{eq:eq-4-4}
\end{equation}

\noindent This choice of $y_{n}$ guarantees the uniqueness of the
representation ${\mathcal{P}}_{\mathcal{A}}S=[y_{i}-y_{j}]$ of the
matrix ${\mathcal{P}}_{\mathcal{A}}S$ defined by relations (\ref{eq:eq-4-2}).
It is clear that ${\mathcal{P}}_{\mathcal{A}}S=S$ if and only if
$S=A$ for some $A\in{\mathcal{A}}_{n}.$ Thus, Lemma \ref{lem:lem-4-1}
together with (\ref{eq:eq-4-4}) shows that the matrix $A=\left[a_{ij}\right]$
belongs to ${\mathcal{A}}_{n}$ if and only if 
\begin{equation}
a_{ij}=y_{i}-y_{j}=\frac{1}{\left|\varrho\right|}\sum^{n}_{k=1}\varrho_{k}(a_{ik}-a_{jk})\ \ (i,j=1,\dots,n).\label{eq:eq-4-5}
\end{equation}

\noindent\begin{flushleft}
These conditions can be rewritten in the form $A=B\ominus B^{T},$
where
\begin{equation}
B\coloneqq\left[\begin{array}{ccc}
y_{1} & \cdots & y_{1}\\
\vdots & \ddots & \vdots\\
y_{n} & \cdots & y_{n}
\end{array}\right]=\sum^{n}_{k=1}y_{k}e_{k}\otimes e\label{eq:eq-4-6}
\end{equation}
and
\begin{equation}
B^{T}=\left[\begin{array}{ccc}
y_{1} & \cdots & y_{n}\\
\vdots & \ddots & \vdots\\
y_{1} & \cdots & y_{n}
\end{array}\right]=\sum^{n}_{k=1}y_{k}e\otimes e_{k}.\label{eq:eq-4-7}
\end{equation}
\par\end{flushleft}

\noindent\begin{flushleft}
Here $e_{1},\dots,e_{n}$ is the standard basis of ${\mathbb{R}}^{n}$
and 
\[
e=e_{1}+\dots+e_{n}=\left(1,\dots,1\right)^{T}.
\]
\par\end{flushleft}
\begin{theorem}
\noindent\label{th:every-additively-consistent-matrix-4-2}Every
additively consistent matrix $A=[y_{i}-y_{j}]$ in ${\mathcal{A}}_{n}$
can be represented as a linear combination 
\begin{equation}
A=\sum^{n-1}_{k=1}\left(y_{k}-y_{n}\right)A_{k}\label{eq:eq-4-8}
\end{equation}

of additively consistent basis matrices
\begin{equation}
A_{k}=e_{k}\otimes e-e\otimes e_{k}{=e}_{k}\ominus e_{k}\ \ \ \left(k=1,\dots,n-1\right)\label{eq:eq-4-9}
\end{equation}

where $y_{1},\dots,y_{n}$ are generalized arithmetic means of rows
of $A.$ Additionally, the set $\{A_{1},\ldots,A_{n-1}\}$ is a basis
of $\mathcal{A}_{n}$.
\end{theorem}
\begin{proof}
In view of (\ref{eq:eq-4-4}) - (\ref{eq:eq-4-7}) we directly obtain
\[
A=B-B^{T}=\sum^{n}_{k=1}y_{k}\left(e_{k}\otimes e-e\otimes e_{k}\right).
\]

Hence expansion (\ref{eq:eq-4-8}) is a consequence of the following
identity
\[
\sum^{n}_{k=1}\left(e_{k}\otimes e-e\otimes e_{k}\right)=\left(\sum^{n}_{k=1}e_{k}\right)\otimes e-e\otimes\left(\sum^{n}_{k=1}e_{k}\right)=
\]
\[
=e\otimes e-e\otimes e=0.
\]

Now we prove the additive consistency of the matrices $A_{k}=e_{k}\otimes e-e\otimes e_{k}.$
For this purpose, we check the equalities ${\mathcal{P}}_{\mathcal{A}}A_{k}=A_{k}.$
By applying twice Lemma \ref{lem:lem-4-1} we obtain 
\[
\mathcal{P}_{\mathcal{A}}\left(e_{k}\otimes e\right)=\left[z_{i}-z_{j}\right]\ \ \ \ (z_{k}=1,z_{i}=0\ \mathrm{for}\ i\neq k)
\]

and 
\[
\mathcal{P}_{\mathcal{A}}\left(e\otimes e_{k}\right)=\left[z_{i}-z_{j}\right]\ \ \ \left(z_{i}=\frac{\varrho_{i}}{\left|\varrho\right|}\,\text{for every}\,i\right).
\]
Hence, we obtain

\noindent
\[
\mathcal{P}_{\mathcal{A}}A_{k}=\mathcal{P}_{\mathcal{A}}\left(e_{k}\otimes e\right)-\mathcal{P}_{\mathcal{A}}\left(e\otimes e_{k}\right)=e_{k}\otimes e-e\otimes e_{k}=A_{k}.
\]

\noindent Finally, the matrices $A_{1},\dots,A_{n-1}$ are linearly
independent. Indeed, if

\noindent
\[
{\alpha}_{1}A_{1}+\dots+{\alpha}_{n-1}A_{n-1}=0,
\]

\noindent then we have

\noindent
\[
\sum^{n-1}_{k=1}\alpha_{k}\mathcal{P}_{\mathcal{A}}A_{k}=\sum^{n-1}_{k=1}\alpha_{k}\left(e_{k}\otimes e-e\otimes e_{k}\right)=\alpha\otimes e-e\otimes\alpha=0,
\]

\noindent where $\alpha={({\alpha}_{1},\dots,{\alpha}_{n-1},0)}^{T}.\ $
Since it is possible only if $\alpha=0,$ the proof is completed.
\end{proof}
\begin{remark}
\noindent\label{rem:remark-4-3}By (\ref{eq:eq-4-9}) it follows
directly that entries of the basis matrices $A_{k}=[a^{k}_{ij}]$
of ${\mathcal{A}}_{n}$ are equal to

\[
a^{k}_{ij}=\begin{cases}
1 & \text{if}\,i=k\,\text{and}\,j\neq k,\\
-1 & \text{if}\,i\neq k\,\text{and}\,j=k,\\
0 & \text{otherwise}.
\end{cases}
\]

\noindent Hence, each $\text{supp}(A_{k})$ has the constant cardinality
$2n-2.$
\end{remark}
Let us consider the following examples. 
\begin{example}
\noindent Suppose that $n=3.$ Then the basis $\{A_{1},A_{2}\}$ of
${\mathcal{A}}_{\boldsymbol{3}}$\textbf{ }is

\noindent
\[
A_{1}=\left[\begin{array}{ccc}
0 & 1 & 1\\
-1 & 0 & 0\\
-1 & 0 & 0
\end{array}\right],\ \ A_{2}=\left[\begin{array}{ccc}
0 & -1 & 0\\
1 & 0 & 1\\
0 & -1 & 0
\end{array}\right].
\]

\noindent If we suppose in addition that the space ${\mathbb{R}}^{3\times3}$
is equipped with the Frobenius inner product, then the Gram-Schmidt
process can be applied to the matrices $A_{1},A_{2}$ and an orthogonal
basis can be computed for ${\mathcal{A}}_{3}:$

\noindent
\[
B_{1}=\left[\begin{array}{ccc}
0 & 1 & 1\\
-1 & 0 & 0\\
-1 & 0 & 0
\end{array}\right],\ \ B_{2}=\left[\begin{array}{ccc}
0 & -{{\frac{1}{2}}} & {{\frac{1}{2}}}\\
{{\frac{1}{2}}} & 0 & 1\\
-{{\frac{1}{2}}} & -1 & 0
\end{array}\right].
\]

\noindent In this basis, the projection $P_{\mathcal{A}}$ is given
by the formula

\noindent
\[
P_{\mathcal{A}}S=\frac{\left\langle S,B_{1}\right\rangle }{\left\langle B_{1},B_{1}\right\rangle }\ B_{1}+\frac{\left\langle S,B_{2}\right\rangle }{\left\langle B_{2},B_{2}\right\rangle }\ B_{2}\ \ \ \left(S\in\mathcal{S}_{3}\right),
\]

\noindent in which coefficients are equal to: 
\[
\frac{\left\langle S,B_{1}\right\rangle }{\left\langle B_{1},B_{1}\right\rangle }=\frac{S_{12}+S_{13}}{2}\,\,\text{and}\,\,\frac{\left\langle S,B_{2}\right\rangle }{\left\langle B_{2},B_{2}\right\rangle }=\frac{-S_{12}+S_{13}+2S_{23}}{3},
\]

\noindent where $s_{ij}$ are elements of the matrix $S.$
\end{example}
\begin{example}
\noindent If $n=4$ then the basis $\{A_{1},A_{2},A_{3}\}$ of $\mathcal{A}_{4}$
is 
\[
A_{1}=\left[\begin{array}{cccc}
0 & 1 & 1 & 1\\
-1 & 0 & 0 & 0\\
-1 & 0 & 0 & 0\\
-1 & 0 & 0 & 0
\end{array}\right],\,\,\,A_{2}=\left[\begin{array}{cccc}
0 & -1 & 0 & 0\\
1 & 0 & 1 & 1\\
0 & -1 & 0 & 0\\
0 & -1 & 0 & 0
\end{array}\right],\,\,\,A_{3}=\left[\begin{array}{cccc}
0 & 0 & -1 & 0\\
0 & 0 & -1 & 0\\
1 & 1 & 0 & 1\\
0 & 0 & -1 & 0
\end{array}\right],
\]

and after orthogonalization
\[
B_{1}=\left[\begin{array}{cccc}
0 & 1 & 1 & 1\\
-1 & 0 & 0 & 0\\
-1 & 0 & 0 & 0\\
-1 & 0 & 0 & 0
\end{array}\right],\,\,\,B_{2}=\left[\begin{array}{cccc}
0 & -\frac{2}{3} & \frac{1}{3} & \frac{1}{3}\\
\frac{2}{3} & 0 & 1 & 1\\
-\frac{1}{3} & -1 & 0 & 0\\
-\frac{1}{3} & -1 & 0 & 0
\end{array}\right],\,\,\,B_{3}=\left[\begin{array}{cccc}
0 & 0 & -\frac{1}{2} & \frac{1}{2}\\
0 & 0 & -\frac{1}{2} & \frac{1}{2}\\
\frac{1}{2} & \frac{1}{2} & 0 & 1\\
-\frac{1}{2} & -\frac{1}{2} & -1 & 0
\end{array}\right].
\]

\noindent For comparison, the basis $\{\widehat{A}_{1},\widehat{A}_{2},\widehat{A}_{3}\}$
of ${\mathcal{A}}_{\boldsymbol{n}}$, discovered by Koczkodaj and
Orłowski \citep{Koczkodaj1997aobf}, is defined as follows: $\widehat{A}_{1}=A_{1}$,
\[
\widehat{A}_{2}=\left[\begin{array}{cccc}
0 & 0 & 1 & 1\\
0 & 0 & 1 & 1\\
-1 & -1 & 0 & 0\\
-1 & -1 & 0 & 0
\end{array}\right],\,\,\,\widehat{A}_{3}=\left[\begin{array}{cccc}
0 & 0 & 0 & 1\\
0 & 0 & 0 & 1\\
0 & 0 & 0 & 1\\
-1 & -1 & -1 & 0
\end{array}\right].
\]
 Since the matrix $\widehat{A}_{2}$ has $8$ nonzero entries, it
follows that it does not have the minimal support of cardinality $6$. 
\end{example}

\section{Orthogonal basis of $\mathcal{A}_{\boldsymbol{n}}$\label{sec:Orthogonal-basis-of}}

\noindent Our next results are concerned with construction of an orthogonal
basis $\{B_{1},\ldots,B_{n-1}\}$ of additively consistent subspace
$\mathcal{A}_{n}=\text{span}\{A_{1},\ldots,A_{n-1}\}$ of ${\mathbb{R}}^{n\times n}$
for any $n>2.$ For simplicity, we confine our attention to the most
important case of the standard Frobenius inner product of the matrices
$A=[a_{ij}]$ and $B=[b_{ij}]$ in ${\mathbb{R}}^{n\times n}$ defined
by 
\[
\left\langle A,B\right\rangle =\sum^{n}_{i,j=1}a_{ij}b_{ij}.
\]
 In other words, throughout this section we assume that $\{B_{1},\ldots,B_{n-1}\}$
is a basis of ${\mathcal{A}}_{\boldsymbol{n}}$ defined by the Gram-Schmidt
recurrent formulae: 
\[
B_{1}=A_{1},
\]

\[
B_{k}=A_{k}-\sum^{k-1}_{i=1}\frac{\left\langle A_{k},B_{i}\right\rangle }{\left\langle B_{i},B_{i}\right\rangle }B_{i}\,\,\text{for}\,\,k=2,\ldots,n-1,
\]

where the matrices $A_{k}$ are defined as in Theorem \ref{th:every-additively-consistent-matrix-4-2}:
\[
A_{k}=e_{k}\otimes e-e\otimes e_{k}=e_{k}\ominus e_{k}\,\,\,k=1,\ldots n-1.
\]

\begin{lemma}
\label{lem:lemma-5-1}It holds that 
\[
\left\langle A_{k},A_{j}\right\rangle =\begin{cases}
2 & \text{if}\,\,k=j\\
2(n-1) & \text{if}\,\,k\neq j
\end{cases}.
\]
\end{lemma}
\begin{proof}
The formula is a direct consequence of Remark \ref{rem:remark-4-3}. 
\end{proof}
\begin{lemma}
\label{lem:lem-5-2}The matrices $A_{k}$ and $B_{k}$ satisfy the
formulae: 
\[
B_{k}=A_{k}+\frac{1}{n-k+1}\ \left(A_{1}+\dots+A_{k-1}\right)
\]
\[
=A_{k}+\frac{1}{n-1}\ B_{1}+\dots+\frac{1}{n-k+1}\ B_{k-1},
\]
\[
\left\langle A_{k},B_{j}\right\rangle =-\frac{2n}{n-j+1}\ \ \ \left(j=1,\dots,k-1\right),
\]
\[
\left\langle B_{k},B_{k}\right\rangle =\frac{2n(n-k)}{n-k+1}
\]

for every $k=1,\dots,n-1.$
\end{lemma}
\begin{proof}
If $k=1,$ then $B_{1}=A_{1}$ and $\left\langle B_{1},B_{1}\right\rangle =\left\langle A_{1},A_{1}\right\rangle .$
It follows from Lemma \ref{lem:lemma-5-1} that $\left\langle B_{1},B_{1}\right\rangle =2\left(n-1\right),$
which finishes the proof. On the other hand, if $1<k<n-1,$ then we
suppose that the lemma is true for $1,\dots,k-1$ and use Lemma \ref{lem:lemma-5-1}
to obtain 
\[
\left\langle A_{k},B_{j}\right\rangle =\left\langle A_{k},A_{j}+\frac{1}{n-j+1}\ \left(A_{1}+\dots+A_{j-1}\right)\right\rangle 
\]
\[
=-2+\frac{-2\left(j-1\right)}{n-j+1}=\frac{-2n}{n-j+1}
\]

and 
\[
\left\langle B_{k},B_{k}\right\rangle =\left\langle A_{k},A_{k}\right\rangle +\frac{2}{n-k+1}\left\langle A_{k},A_{1}+\dots+A_{k-1}\right\rangle 
\]
\[
+\frac{1}{{(n-k+1)}^{2}}\left\langle A_{1}+\dots+A_{k-1},A_{1}+\dots+A_{k-1}\right\rangle 
\]
\[
=2\left(n-1\right)-\frac{4\left(k-1\right)}{n-k+1}+\frac{2\left(k-1\right)\left(n-1\right)-2(k-1)(k-2)}{{(n-k+1)}^{2}}
\]
 {}
\[
=\frac{2n(n-k)}{n-k+1}.
\]

\noindent Hence, it follows from the Gram-Schmidt formulae that

\noindent
\[
B_{k}=A_{k}-\frac{\left\langle A_{k},B_{1}\right\rangle }{\left\langle B_{1},B_{1}\right\rangle }\ B_{1}-\dots-\frac{\left\langle A_{k},B_{k-1}\right\rangle }{\left\langle B_{k-1},B_{k-1}\right\rangle }\ B_{k-1}
\]

\[
=A_{k}+\frac{1}{n-1}\ B_{1}+\dots+\frac{1}{n-k+1}\ B_{k-1}.\ \ \ 
\]
These formulae allow us to finish the inductive step of the proof:

\noindent
\[
B_{k}=A_{k}+\frac{1}{n-k+2}\ \left(A_{1}+\dots+A_{k-2}\right)
\]

\[
+\frac{1}{n-k+2}\ \left(A_{k-1}+{\frac{1}{n-k+2}(A}_{1}+\dots+A_{k-2})\right)
\]

\[
=A_{k}+\frac{1}{n-k+1}\ \left(A_{1}+\dots+A_{k-1}\right).
\]

\noindent This completes the proof.
\end{proof}
\noindent The last lemma is a crucial step to derive closed formulae
for the orthogonal basis $\mathrm{\{}$$B_{1},\dots,B_{n-1}\}$ of
$\mathcal{A}_{n}.$ Indeed, let $L:{\mathbb{R}}^{n}\longrightarrow{\mathbb{R}}^{n\times n}$
be a linear mapping defined by 
\[
L(x)=x\otimes e-e\otimes x=x\ominus x,\,\,\,x=\left(x_{1},\dots,x_{n}\right)^{T}\in\mathbb{R}^{n}.
\]
Then we can write the basis $A_{k}=e_{k}\otimes e-e\otimes e_{k}$
of ${\mathcal{A}}_{n}$ in the form

\[
A_{k}=L(e_{k})\,\,\,\text{for}\,\,\,\left(k=1,\dots,n-1\right).
\]

\begin{theorem}
\noindent\label{thm:thm-5-3}Let $x^{k}={\left(x^{k}_{1},\dots,x^{k}_{n}\right)}^{T}$
be vectors in ${\mathbb{R}}^{n}$ such that 
\[
x^{k}_{1}=\dots=x^{k}_{k-1}=\frac{1}{n-k+1},\ \ x^{k}_{k}=1,\ \ x^{k}_{k+1}=\dots=x^{k}_{n}=0.
\]
Then 
\[
B_{k}=x^{k}\otimes e-e\otimes x^{k}=x^{k}\ominus\ x^{k}\ \left(k=1,\dots,n-1\right).
\]
\end{theorem}
\begin{proof}
By Lemma \ref{lem:lem-5-2} it follows that 
\[
B_{k}=A_{k}+\frac{1}{n-k+1}\ \left(A_{1}+\dots+A_{k-1}\right)
\]
\[
=L(e_{k})+\frac{1}{n-k+1}\ \left(L(e_{1})+\dots+L(e_{k-1})\right)
\]
\[
=L(e_{k}+\frac{1}{n-k+1}\left(e_{1}+\dots+e_{k-1}\right))=L\left(x^{k}\right),
\]
which completes the proof.
\end{proof}
\begin{corollary}
\textbf{\label{cor:For-any-skew-symmetric-5-4}}For any skew-symmetric
matrix $S$ in ${\mathbb{R}}^{n\times n}$, we have 
\[
{\mathcal{P}}_{\mathcal{A}}S=\frac{1}{2n}\sum^{n-1}_{k=1}\frac{(n-k+1)\left\langle S,B_{k}\right\rangle }{n-k}B_{k}.
\]
\end{corollary}
\begin{proof}
Since $\mathrm{\{}$$B_{1},\dots,B_{n-1}\}$ is an orthogonal basis
of ${\mathcal{A}}_{n},$ we have 
\[
\mathcal{P}_{\mathcal{A}}S=\sum^{n-1}_{k=1}\frac{\left\langle S,B_{k}\right\rangle }{\left\langle B_{k},B_{k}\right\rangle }B_{k}.
\]
Hence the corollary follows from Theorem \ref{thm:thm-5-3}.
\end{proof}
\begin{example}
In the case ${\mathbb{R}}^{4\times4}$ vectors $x^{k}$ are equal
to:
\[
x^{1}=\left(1,0,0,0\right)^{T},x^{2}=\left(\frac{1}{3},1,0,0\right)^{T},x^{3}=\left(\frac{1}{2},\frac{1}{2},1,0\right)^{T}.
\]
Hence, one obtains the orthogonal basis $B_{k}=x^{k}\ominus x^{k}$
of $\mathcal{A}_{4}$: 
\[
B_{1}=\left[\begin{array}{ccc}
0 & 1 & \begin{array}{cc}
1 & 1\end{array}\\
-1 & 0 & \begin{array}{cc}
0 & 0\end{array}\\
\begin{array}{c}
-1\\
-1
\end{array} & \begin{array}{c}
0\\
0
\end{array} & \begin{array}{c}
\begin{array}{cc}
0 & 0\end{array}\\
\begin{array}{cc}
0 & 0\end{array}
\end{array}
\end{array}\right],
\]
\[
B_{2}=\left[\begin{array}{ccc}
0 & -\frac{2}{3} & \begin{array}{cc}
\frac{1}{3} & \frac{1}{3}\end{array}\\
\frac{2}{3} & 0 & \begin{array}{cc}
1 & 1\end{array}\\
\begin{array}{c}
-\frac{1}{3}\\
-\frac{1}{3}
\end{array} & \begin{array}{c}
-1\\
-1
\end{array} & \begin{array}{c}
\begin{array}{cc}
0 & 0\end{array}\\
\begin{array}{cc}
0 & 0\end{array}
\end{array}
\end{array}\right]\ ,\ \ B_{3}=\left[\begin{array}{ccc}
0 & 0 & -\begin{array}{cc}
\frac{1}{2} & \frac{1}{2}\end{array}\\
0 & 0 & \begin{array}{cc}
-\frac{1}{2} & \frac{1}{2}\end{array}\\
\begin{array}{c}
\frac{1}{2}\\
-\frac{1}{2}
\end{array} & \begin{array}{c}
\frac{1}{2}\\
-\frac{1}{2}
\end{array} & \begin{array}{c}
\begin{array}{cc}
0 & 1\end{array}\\
\begin{array}{cc}
-1 & 0\end{array}
\end{array}
\end{array}\right].
\]
By Theorem \ref{thm:thm-5-3} it follows directly that 
\[
\mathcal{P}_{\mathcal{A}}S=\frac{\left\langle S,B_{1}\right\rangle }{\left\langle B_{1},B_{1}\right\rangle }\ B_{1}+\frac{\left\langle S,B_{2}\right\rangle }{\left\langle B_{2},B_{2}\right\rangle }\ B_{2}+\frac{\left\langle S,B_{3}\right\rangle }{\left\langle B_{3},B_{3}\right\rangle }\ B_{3}\ \ \ (S\in\mathcal{S}_{4})
\]
where 
\[
\frac{\left\langle S,B_{1}\right\rangle }{\left\langle B_{1},B_{1}\right\rangle }=\frac{s_{12}+s_{13}+s_{14}}{3},
\]
\[
\frac{\left\langle S,B_{2}\right\rangle }{\left\langle B_{2},B_{2}\right\rangle }=\frac{{-2s}_{12}+s_{13}+s_{14}+3s_{23}+3s_{24}}{4},
\]
and 
\[
\frac{\left\langle S,B_{3}\right\rangle }{\left\langle B_{3},B_{3}\right\rangle }=\frac{-s_{12}+s_{13-}s_{23}+3s_{24}}{2}.
\]
\end{example}

\section{Orthogonal windowing of PC matrices\label{sec:Orthogonal-windowing-of}}

\noindent The classical definition of the discrete Fourier transform
$\mathbb{R}^{n}\ni x\longrightarrow\mathcal{F}(x)\in\mathbb{R}^{n}$
asserts that the i-th component of $\mathcal{F}(x)$ is equal to:
\[
\mathcal{F}_{i}\left(x\right)=\sum^{n}_{k=1}x_{k}\omega^{\left(i-1\right)\left(k-1\right)}\,\,\,\text{for}\,\,\,\left(1\le i\le n\right),
\]

\noindent where $\omega$ is a principal n-th root of unity. It is
an orthogonal projection whose values are equal to coefficients of
the interpolating polynomial with nodes at the n-th roots $1,\omega,\dots,{\omega}^{n-1}$
of unity.

\noindent This definition has inspired us to introduce the following
projection $\mathcal{P}_{\mathcal{C}}:$ $\mathcal{M}_{n}\longrightarrow\mathcal{C}_{n}$
of the abelian group $\mathcal{M}_{n}=(\mathcal{M}_{n},\cdot)$ of
$PC$ matrices in $\mathbb{R}^{n\times n}$ onto the consistent subgroup
$\mathcal{C}_{n}=\left(\mathcal{C}_{n},\cdot\right),$ where the dot
denotes the Hadamard product of matrices.
\begin{definition}
\noindent\label{def:def-6-1}Le\textbf{t }$M=[m_{ij}]$ be a $PC$
matrix in ${\mathbb{R}}^{n\times n}.$ If $\varrho=\left({\varrho}_{i}\right)$
is a vector of $n$ positive weights, then the $PC$ projection $\mathcal{P}_{\mathcal{C}}M=[r_{ij}]$
of $M$ onto $\mathcal{C}_{n}$ is a consistent matrix having elements
equal to: 
\begin{equation}
r_{ij}=\left[\prod^{n}_{k=1}\left(\frac{m_{ik}}{m_{jk}}\right)^{\varrho_{k}}\right]^{1/\left|\varrho\right|}\ \ \ \ (1\le i,j\le n).\label{eq:eq-6-1}
\end{equation}
The domain of $\mathcal{P}_{\mathcal{C}}$ can be extended to all
positive matrices. There is a close relationship between the $PC$
projection $\mathcal{P}_{\mathcal{C}}$ and orthogonal projection
$\mathcal{P}_{\mathcal{A}}$ investigated in Theorem~\ref{thm:thm-5-3}.
\end{definition}
\begin{lemma}
We have \label{lem:lem-6-2} 
\[
\mathcal{P}_{\mathcal{C}}M=\exp(\mathcal{P}_{\mathcal{A}}(\ln M))
\]
for every $M\in\mathcal{M}_{n}.$
\end{lemma}
\begin{proof}
\noindent In view of Lemma \ref{lem:lem-4-1} the matrix $\mathcal{P}_{\mathcal{A}}\left(\ln M\right)=[a_{ij}]$
has elements equal to

\noindent
\[
a_{ij}=\frac{1}{\left|\varrho\right|}\sum^{n}_{k=1}\varrho_{k}(\ln m_{ik}-\ln m_{jk})
\]

\[
=\ln\left(\left[\prod^{n}_{k=1}\left(\frac{m_{ik}}{m_{jk}}\right)^{\varrho_{k}}\right]^{1/\left|\varrho\right|}\right).
\]
Hence, it remains to compute the exponential function of both sides.
\end{proof}
\noindent Now we formulate the method of $PC$ windowing by an analogy
of the \emph{FFT} windowing in Fourier analysis. For this purpose,
let the inner power of the matrices $M=[m_{ij}]$ and $B=[b_{ij}]$
in ${\mathbb{R}}^{n\times n}$ be defined by

\noindent
\begin{equation}
\left[M,B\right]=\prod^{n}_{i,j=1}(m_{ij})^{b_{ij}}.\label{eq:eq-6-2}
\end{equation}

\noindent This definition helps to simplify our description of orthogonal
decomposition for the $PC$ projection ${\mathcal{P}}_{\mathcal{C}}:{\mathcal{M}}_{n}\longrightarrow{\mathcal{C}}_{n}.$
As in the previous section, we assume that the basis $\mathrm{\{}$$B_{1},\dots,B_{n-1}\}$
of ${\mathcal{A}}_{n}$ is orthogonal with respect to the standard
Frobenius inner product in ${\mathbb{R}}^{n\times n}.$ It guarantees
validity of explicit formulae for $B_{k}$ that are presented in Theorem
\ref{thm:thm-5-3} and Corollary \ref{cor:For-any-skew-symmetric-5-4}.
\begin{theorem}
\noindent\label{thm:For-any-matrix-6-2}For any matrix $M=[m_{ij}]$
in ${\mathcal{M}}_{n}$, we have

\noindent
\[
\mathcal{P}_{\mathcal{C}}M=\prod^{n-1}_{k=1}[M,B_{k}]^{\ \frac{n-k+1}{2n(n-k)}B_{k}}.
\]
\end{theorem}
\begin{proof}
\noindent By using Lemma \ref{lem:lem-6-2} and Corollary \ref{cor:For-any-skew-symmetric-5-4}
we readily obtain 
\[
\mathcal{P}_{\mathcal{C}}M=\exp(\mathcal{P}_{\mathcal{A}}(\ln M))=\mathrm{exp\ }\left(\sum^{n-1}_{k=1}\frac{\left(n-k+1\right)\left\langle \ln M,B_{k}\right\rangle }{2n(n-k)}B_{k}\right).
\]

\noindent Thus the matrix $\mathcal{P}_{\mathcal{C}}M=[r_{ij}]$ has
elements equal to

\noindent
\[
r_{ij}=\exp\left(\sum^{n-1}_{k=1}\frac{(n-k+1)b^{\left(k\right)}_{ij}}{2n(n-k)}\sum^{n}_{u,v=1}b^{(k)}_{uv}\ln m_{uv}\right),
\]

\noindent where $B_{k}=\left[b^{\left(k\right)}_{ij}\right]$.

\noindent Hence it follows from the definition of the inner power
$[\cdot,\cdot]$ that 
\[
r_{ij}=\exp\left(\sum^{n-1}_{k=1}\frac{(n-k+1)b^{\left(k\right)}_{ij}}{2n(n-k)}\mathrm{ln}[M,B_{k}]\right)=\prod^{n-1}_{k=1}[M,B_{k}]^{c^{\left(k\right)}_{ij}}
\]

\noindent where $c^{\left(k\right)}_{ij}=\frac{n-k+1}{2n(n-k)}b^{\left(k\right)}_{ij}$.
So this is equivalent to our assertion.
\end{proof}

\section{Saaty windowing of PC matrices\label{sec:Saaty-windowing-of}}

\noindent The celebrated Saaty's AHP \citep{Saaty1977asmf,Saaty1980tahp}
is based on the Perron theorem \citep{Gantmaher2000ttom,Perron1907ctdm}
which asserts that a positive matrix $M$ in ${\mathbb{R}}^{n\times n}$
has a positive eigenvalue $\lambda_{max}\left(M\right)$\textit{ }which
is a simple root of the characteristic equation $\det\left(M-\lambda I\right)=0$
and exceeds the moduli of all the other eigenvalues. Additionally,
to this maximal eigenvalue $\lambda_{max}>0$ there is a corresponding
positive eigenvector $x^{+}=\left(x^{+}_{1},\dots,\ x^{+}_{n}\right)^{T}>0$
which is unique, up to a positive factor.
\begin{remark}
\noindent\label{rem:rem-7-1}The Perron pair $(\lambda_{max},x^{+})$
of $M>0$ is strongly unique in the set of all eigenpairs $\left(\lambda,y\right)$
such that $My=\lambda y\ \left(y\neq0\right).$ This means that if
$\left(\lambda,y\right)$ is an eigenpair of $M$ such that $\left|\lambda\right|<{\lambda}_{max},$
then it is not true that either $y\ge0$ or $y\le0.$ Indeed, let
$z>0$ denote an eigenvector of the transpose matrix $M^{T}$ corresponding
to ${\lambda}_{max}.$ Then 
\[
\lambda_{max}\left\langle y,z\right\rangle =\left\langle y,M^{T}\mathrm{z}\right\rangle =\left\langle My,z\right\rangle =\lambda\left\langle y,z\right\rangle 
\]

\noindent and so $\left\langle y,z\right\rangle =0.$ This is possible
only when the vector $y\neq0$ has both negative and positive components.
\end{remark}
\begin{definition}
\noindent\textbf{\label{def:def-7-2}} If $({\lambda}_{max},x^{+})$
denotes a Perron pair of a matrix $M$ in the group $\left(\mathcal{M}_{n},\cdot\right)$,
then the consistent projection 
\[
\mathcal{R}_{\mathcal{C}}M=x^{+}\oslash x^{+}\ \ \ (M\in\mathcal{M}_{n})
\]

\noindent is said to be the Saaty projection of ${\mathcal{M}}_{n}$
onto ${\mathcal{C}}_{n}.$ Note that one can extend the domain of
${\mathcal{R}}_{\mathcal{C}}$ to all positive matrices $M$ in ${\mathbb{R}}^{n\times n}.$
\end{definition}
\noindent The Saaty projection ${\mathcal{R}}_{\mathcal{C}}$ is widely
used in applications of AHP in many areas, including psychology, product
management, strategic planning, finance and banking, market research,
and others; see e.g. Saaty \citep{Saaty1980tahp}, Golden et al. \citep{Golden1989tahp},
and Sipahi et al. \citep{Sipahi2010tahp}. This is still the case
in spite of criticism of the eigenvector approach to pairwise comparisons
found in the papers of Blanquero et al. \citep{Blanquero2006iewf},
Gass and Rapcsak \citep{Gass2004svdi}, Bozoki and Rapcsak \citep{Bozoki2008osak},
Gass and Standard \citep{Gass2002copr}, and the references therein.
An additional disadvantage of this technique is provided in the next
example.
\begin{example}
\noindent Let $P=\left[p_{ij}\right]$ be a stochastic matrix in ${\mathbb{R}}^{n\times n}.$
This means that $P>0$ and the sum of each row of $P$ is $1$, i.e.,

\noindent
\[
\sum^{n}_{j=1}p_{ij}=1\ \ \ \left(i=1,\dots,n\right).
\]

\noindent Hence $P$ has the eigenvalue $1$ with the positive eigenvector
$e=\left(1,\dots,1\right)\in\mathbb{R}^{n}.$ By Remark \ref{rem:rem-7-1}
it follows that $(1,e)$ is a Perron pair of $P.$ Thus, Saaty's projection
${\mathcal{R}}_{\mathcal{C}}P=e\oslash e\ $ has all elements equal
to $1,$ and so is independent of $P.$ In other words, Saaty's projection
does not distinguish between any two stochastic matrices.
\end{example}
\noindent The last example shows that the usefulness of the Saaty
projection ${\mathcal{R}}_{\mathcal{C}}$ may be diminished for those
positive matrices which are close to stochastic matrices. On the other
hand, it suggests the following characterization of ${\mathcal{R}}_{\mathcal{C}}$
as a Hausdorff quotient of the scaled matrix ${\lambda}^{-1}_{max}M$
by the stochastic matrix $P$.
\begin{proposition}
\noindent\label{prop:prop-7-3}For every positive matrix $M$ in
${\mathbb{R}}^{n\times n},$ there exists a stochastic matrix $P$
such that 
\begin{equation}
\mathcal{R}_{\mathcal{C}}M=\frac{\lambda^{-1}_{max}M}{P}.\label{eq:eq-7-1}
\end{equation}
\end{proposition}
\begin{proof}
\noindent Let (${\lambda}_{max},x^{+})$ be the Perron pair of $M.$
Then it follows from Definition \ref{def:def-7-2} that the matrix
identity (\ref{eq:eq-7-1}) is equivalent to 
\[
\frac{x^{+}_{i}}{x^{+}_{j}}=\lambda^{-1}_{max}\frac{m_{ij}}{p_{ij}}\ \ \ \left(i,j=1,\dots,n\right),
\]

\noindent which defines the matrix $P=\left[p_{ij}\right]>0.$ Since
$Mx^{+}=\lambda_{max}x^{+},$ we have

\noindent
\[
\sum^{n}_{j=1}p_{ij}=\frac{\mathrm{\ 1}}{\lambda_{max}x^{+}_{i}}\sum^{n}_{j=1}m_{ij}x^{+}_{j}=\frac{\lambda_{max}x^{+}_{i}}{\lambda_{max}x^{+}_{i}}=1\,\,\,\text{for}\,\,\,\left(i=1,\dots,n\right),
\]

\noindent which completes the proof.
\end{proof}
\noindent Now we return to study the $PC$ projection $\mathcal{P}_{\mathcal{C}}=\exp(\mathcal{P}_{\mathcal{A}}(\ln M))$,
where $\mathcal{P}_{\mathcal{A}}$ is the orthogonal projection onto
the subspace $\mathcal{A}_{n}$ of space $\mathbb{R}^{n\times n}$
endowed with the weighted Frobenius inner product $\left\langle \cdot,\cdot\right\rangle _{F(\varrho)}$.
By Definition \ref{def:def-6-1} we have ${\mathcal{P}}_{\mathcal{C}}M=\left[r_{ij}\right],$
where

\noindent
\begin{equation}
r_{ij}=\left[\prod^{n}_{k=1}\left(\frac{m_{ik}}{m_{jk}}\right)^{\varrho_{k}}\right]^{1/\left|\varrho\right|}\ \ \ \ (1\le i,j\le n).\label{eq:eq-7-2}
\end{equation}

\noindent In order to compare ${\mathcal{P}}_{\mathcal{C}}$ with
${\mathcal{R}}_{\mathcal{C}},$ we suppose that the domain of ${\mathcal{P}}_{\mathcal{C}}$
is extended to all positive matrices $M$ in ${\mathbb{R}}^{n\times n}$
and consider first the important special case when ${\varrho}_{k}=1$
for every $k.$
\begin{example}
\noindent Let $Q=\left[q_{ij}\right]$ be a multiplicatively stochastic
matrix in ${\mathbb{R}}^{n\times n}.$ This means that $Q>0$ and
the product of each row of $Q$ is $1$, i.e., 
\[
\prod^{n}_{j=1}q_{ij}=1\,\,\,\left(i=1,\dots,n\right).
\]

\noindent Then it follows from \ref{rem:rem-7-1} that the matrix
${\mathcal{P}}_{\mathcal{C}}Q=e\oslash e\ $ has all elements equal
to $1,$ and so is independent of $Q.$ In other words, the $PC$
projection $\mathcal{P}_{\mathcal{C}}$ does not distinguish between
any two multiplicatively stochastic matrices. For the weighted inner
product, the same conclusion is true, whenever each row of $Q$ satisfies
$\frac{1}{\left|\varrho\right|}\sum^{n}_{j=1}\varrho_{j}q_{ij}=0\,\,\,\,\left(i=1,\dots,n\right).$

The last two examples support an interesting result of Golany and
Kres presented in \citep{Golany1993ameo} who compared a few methods
commonly used in pairwise comparisons and concluded that there is
no windowing method that is superior to the other ones in all cases.
Note that this conclusion contradicts the claim of Saaty and Vargas
\citep{Saaty1984coel} that the eigenvector method is superior to
the logarithmic least-squares method.
\end{example}
\begin{remark}
\noindent\label{rem:rem-7-4}In Section \ref{sec:Consistent-approximation-of}
we have briefly discussed how to choose a weighted Frobenius norm
$\left\Vert \cdot\right\Vert _{F(\varrho)}$ in cases of randomly
disturbed elements $m_{ij}$ of a given $PC$ matrix $M.$ It would
be of great interest to undertake numerical experiments on the following
choice of the positive weight vector $\varrho=\left({\varrho}_{1},\dots,{\varrho}_{n}\right)$
in the definition of the logarithmic consistent problem from Section
\ref{sec:Consistent-approximation-of}:

\noindent
\[
\varrho_{1}=x^{+}_{1},\dots,\varrho_{n}=x^{+}_{n},
\]

\noindent where $x^{+}=(x^{+}_{1},\dots,x^{+}_{n})$ is the Perron
eigenvector of $M.$ Some theoretical considerations suggest that
such a choice may be useful in cases when Saaty's eigenvector approach
fails.
\end{remark}

\section{SVD windowing of PC matrices\label{sec:SVD-windowing-of}}

\noindent It is clear that a good windowing of the $PC$ matrix $M>0$
should reflect its structure to the highest possible degree. Therefore,\textbf{
}in this section we propose a new class of windowing. It is inspired
by an interesting paper \citep{Gass2004svdi} of Gass and Rapcsak,
who discovered an original method to capture consistency by using
the celebrated singular value decomposition (SVD) from matrix algebra
\citep{Golub1996mc}. Their method can be viewed as an attempt to
improve the eigenvector method according to Saaty \citep{Saaty1980tahp}.

\noindent In order to define the SVD windowing, we recall the singular
value decomposition of $M$ \citep{Golub1996mc}: 
\begin{equation}
M=U^{T}\mathrm{\Sigma}V\,\,\text{where}\,\,\mathrm{\Sigma}=\text{diag}\left(\sigma_{1},\dots,\sigma_{n}\right),\label{eq:eq-8-1}
\end{equation}
in which $M$ may be any matrix in $\mathbb{R}^{n\times n}$, the
matrices $U=\left[u_{1},\dots,u_{n}\right]$ and $V=[v_{1},\dots,v_{n}]$
are orthogonal, and $\sigma_{1}\ge\dots\ge\sigma_{n}\ge0.$ The columns
$u_{k}$ and $v_{k}$ of $U$ and $V$ are called singular vectors
of $M$ corresponding to the singular value $\sigma_{k}$ for $\left(1\le k\le n\right)$.
They satisfy the following conditions:
\begin{align}
Mv_{k}=\sigma_{k}u_{k}, & M^{T}u_{k}=\sigma_{k}v_{k},\nonumber \\
MM^{T}u_{k}=\sigma^{2}_{k}u_{k}, & M^{T}Mv_{k}=\sigma^{2}_{k}v_{k}\label{eq:eq-8-2}
\end{align}
Additionally, if $M$ is a positive matrix in $\mathbb{R}^{n\times n},$
then it follows from the Perron theorem that ${\sigma}_{1}>{\sigma}_{2}.\ $
\begin{definition}
\noindent\label{def:def-8-1}Le\textbf{t }$M=U^{T}\mathrm{\Sigma}V$
denote the singular value decomposition of $M>0$ in $\mathbb{R}^{n\times n}${\normalsize ,}
where $U=\left[u_{1},\dots,u_{n}\right]$, $V=\left[v_{1},\dots,v_{n}\right]$
and $\mathrm{\Sigma}=\text{diag}\left(\sigma_{1},\dots,\sigma_{n}\right)$.
Then the consistent projection of the group ${\mathcal{M}}_{n}$ onto
${\mathcal{C}}_{n}$ is defined by 
\[
\mathcal{T}_{\mathcal{C}}M=u_{1}\oslash u_{1}\,\,\,\text{for}\,\,\,(M\in\mathcal{M}_{n}).
\]

\noindent This is said to be the SVD projection. 
\end{definition}
Note that the domain of ${\mathcal{T}}_{\mathcal{C}}$ can be extended
to all positive matrices $M$ in ${\mathbb{R}}^{n\times n}.$
\begin{example}
\noindent\label{ex:ex-8-2}Let $M$ be a normal positive matrix in
$\mathbb{R}^{n\times n}.$ By definition of normality $MM^{T}=M^{T}M${\normalsize .}
Hence the matrices $M$ and $M^{T}$ have a common system of orthogonal
eigenvectors ${\ u}_{k}$ \citep[Sec. IX.10]{Gantmaher2000ttom} such
that 
\[
Mu_{k}=\lambda_{k}u_{k}\ ,M^{T}u_{k}=\overline{\lambda_{k}}u_{k}\ ,MM^{T}u_{k}=M^{T}Mu_{k}=\left|\lambda_{k}\right|^{2}u_{k}
\]

\noindent for every $k=1,\dots,n.$ Thus ${\sigma}_{k}=\left|{\lambda}_{k}\right|,$
and so ${\sigma}_{1}={\lambda}_{1}$ and $\mathcal{T}_{\mathcal{C}}M=\mathcal{R}_{\mathcal{C}}M.$
Thus, the eigenvector and SVD windowing coincide for the class of
normal matrices.
\end{example}
\begin{example}
\noindent\label{exa:ex-8-3}Let $P=\left[p_{ij}\right]$ be a doubly
stochastic positive matrix in $\mathbb{R}^{n\times n}$, i.e., $P$
and $P^{T}$ are stochastic positive matrices. Hence we have 
\[
Pe=e,P^{T}e=e,PP^{T}e=P^{T}Pe=e,
\]

\noindent where $e=\left(1,\dots,1\right)\in\mathbb{R}^{n}.$ Therefore,
the Saaty and SVD projections coincide for the positive matrices $M$
and $M^{T}$ if and only if 
\[
\mathcal{R}_{\mathcal{C}}M=\mathcal{R}_{\mathcal{C}}\left(M^{T}\right)=\mathcal{T}_{\mathcal{C}}M=\mathcal{T}_{\mathcal{C}}\left(M^{T}\right)=\alpha e\oslash e
\]

\noindent for some constant $\alpha>0.$ Additionally, if a positive
matrix $P$ is stochastic but not doubly stochastic, then $\mathcal{R}_{\mathcal{C}}P=\mathcal{T}_{\mathcal{C}}P=e\oslash e,\ \ \mathcal{R}_{\mathcal{C}}\left(P^{T}\right)\neq\alpha e\oslash e$
and $\mathcal{T}_{\mathcal{C}}\left(P^{T}\right)\neq\alpha e\oslash e$
for every $\alpha>0.$
\end{example}
\noindent As an addendum to the last example, the important Sinkhorn
theorem \citep{Sinkhorn1964arba} should be mentioned. It asserts
that to each positive matrix $M$ in $\mathbb{R}^{n\times n}$ there
is a unique doubly stochastic positive matrix of the form $D_{1}MD_{2}$,
where the diagonal matrices $D_{1}$ and $D_{2}$ with positive main
diagonals are unique up to a positive scalar factor. The $D_{1}MD_{2}$
matrix is obtained as a limit of the sequence of matrices generated
by alternately normalizing the rows and columns of $M$ so that the
sums of their elements are equal $1$.

\noindent In contrast with the eigenvector projection ${\mathcal{R}}_{\mathcal{C}},$
the SVD projection ${\mathcal{T}}_{\mathcal{C}}$ is not only an eigenvector
based projection but also a distance based projection. The first assertion
follows readily from (\ref{eq:eq-8-2}), while the second one is much
harder to prove. For the proof, we first embed the consistent group
$\mathcal{C}_{n}=\left(\mathcal{C}_{n},\cdot\right)$ into the supergroup
$\mathcal{G}_{n}=\left(\mathcal{G}_{n},\cdot\right)$ of all tensor
products $x\otimes y=xy^{T}$ of positive vectors $x,y\in\mathbb{R}^{n}$,
in which the Hadamard dot product satisfies 
\[
\left(x\otimes y\right)\cdot\left(z\otimes w\right)=\left\langle y,z\right\rangle \left(x\otimes w\right).
\]

\noindent Next, we consider a best approximation problem for $\mathcal{G}_{n}$
of the form 
\[
\left\Vert M-\mathcal{T}_{\mathcal{G}}M\right\Vert _{F}=\mathop{\mathrm{inf}}_{x>0,y>0}\left\Vert M-x\otimes y\right\Vert _{F},
\]
for which the following theorem, according to Eckart, Young \citep{Eckart1936taoo}
and Mirsky \citep{Mirsky1960sgfa}, holds.
\begin{theorem}
\noindent For every $M\in\mathbb{R}^{n\times n}$, we have 
\[
\mathcal{T}_{\mathcal{G}}M=\sigma_{1}u_{1}\otimes v_{1}\,\,\text{and\,\,}\left\Vert M-\mathcal{T}_{\mathcal{G}}M\right\Vert _{F}=\sum^{n}_{k=2}\sigma^{2}_{k}.
\]
\end{theorem}
\noindent This theorem helps to prove that the SVD projection $\mathcal{T}_{\mathcal{C}}$
is a composite projection of two distance based projections: additively
consistent weighted projection ${\mathcal{P}}_{\mathcal{A}}$ and
tensor projection ${\mathcal{T}}_{\mathcal{G}}.$ A conclusion of
this kind is not true for the eigenvector projection ${\mathcal{R}}_{\mathcal{C}}.$

\noindent To the end of this section, the matrices $U=\left[u_{1},\dots,u_{n}\right]$,
$V=\left[v_{1},\dots,v_{n}\right]$ and $\mathrm{\Sigma}=\text{diag}\left(\sigma_{1},\dots,\sigma_{n}\right)$
are understood to be factors of the singular value decomposition $M=U^{T}\mathrm{\Sigma}V$
of $M\in{\mathbb{R}}^{n\times n}.$
\begin{theorem}
\noindent\label{thm:thm-8-2}The SVD projection ${\mathcal{T}}_{\mathcal{C}}$
is a composite projection of the distance based logarithmic and tensor
projections $\mathcal{P}_{\mathcal{C}}$ and $\mathcal{T}_{\mathcal{G}}$
such that 
\[
T_{\mathcal{C}}M=\mathcal{P}_{\mathcal{C}}\left(\mathcal{T}_{\mathcal{G}}\mathrm{\ }\mathrm{M}\right)=u_{1}\oslash u_{1}
\]

\noindent for every PC matrix $M$ in $R^{n\times n}.$
\end{theorem}
\begin{proof}
\noindent If we denote ${\mathcal{T}}_{\mathcal{G}}M=\left[s_{ij}\right],$
then it follows from the Eckart-Young-Mirsky theorem that

\noindent
\[
s_{ij}={\sigma}_{1}u_{1i}v_{1j}\ \ \left(1\le i,j\le n\right),
\]

\noindent where $u_{1}={(u_{11},\dots,u_{1n})}^{T}$ and $v_{1}={(v_{11},\dots,v_{1n})}^{T}.$
Further, by applying Lemma \ref{lem:lem-4-1} and (\ref{eq:eq-4-2})
to the matrix $\mathcal{T}_{\mathcal{G}}M=\left[s_{ij}\right]$ we
conclude that the elements $r_{ij}$ of the matrix $\mathcal{P}_{\mathcal{C}}\left(\mathcal{T}_{\mathcal{G}}\mathrm{M}\right)$
are equal to 
\[
r_{ij}=\left[\prod^{n}_{k=1}(\frac{\sigma_{1}u_{1i}v_{1k}}{\sigma_{1}u_{1j}v_{1k}})^{\varrho_{k}}\right]^{1/\left|\varrho\right|}=\frac{u_{1i}}{u_{1j}}.
\]

\noindent Thus it follows from Definition \ref{def:def-8-1} that

\noindent
\[
{{\mathcal{P}}_{\mathcal{C}}\left({\mathcal{T}}_{\mathcal{G}}\mathrm{\ }\mathrm{M}\right)=u_{1}{\oslash u}_{1}={\mathcal{T}}_{\mathcal{C}}M.\ }
\]

\noindent This completes the proof.
\end{proof}
\noindent The next lemma is crucial for the proof of a new characterization
of the eigenvector and tensor projections. It shows that the Saaty
and SVD windowing forms coincide on the supergroup $G_{n}=\left(G_{n},\cdot\right)$
of tensor products of all positive vectors in ${\mathbb{R}}^{n}.$
Thus, it is another theoretical argument for the conclusion of Golany
and Kres \citep{Golany1993ameo} that neither windowing method is
superior in all cases.
\begin{lemma}
\noindent\label{lem:lem-8-6}If $a,b$ are\textbf{ }positive vectors
in $\mathbb{R}^{n}$, then we have 
\[
\mathcal{T}_{\mathcal{C}}\left(a\otimes b\right)=\mathcal{R}_{\mathcal{C}}\left(a\otimes b\right)=a\oslash a.
\]
\end{lemma}
\begin{proof}
\noindent The matrix $C=a\otimes b=\left[a_{i}b_{j}\right]$ has the
rank $1$. Hence we have 
\[
\det\left(C-\lambda I\right)=(-\lambda)^{n}+\left\langle a,b\right\rangle (-\lambda)^{n-1},
\]

\noindent because all other coefficients of this characteristic polynomial
of $C$ disappear as sums of higher rank minors. Hence $C$ has eigenvalues
$\lambda_{1}=\left\langle a,b\right\rangle $ and $\lambda_{2}=\dots=\lambda_{n}=0.$
Moreover, we have

\noindent
\[
\left(a\otimes b\right)a=a\left(b^{T}a\right)=\left\langle a,b\right\rangle a.
\]

\noindent Thus $a$ is the eigenvector of $a\otimes b$ corresponding
to the largest eigenvalue $\left\langle a,b\right\rangle $. By Definition
\ref{def:def-7-2} it proves the second identity 
\[
\mathcal{R}_{\mathcal{C}}\left(a\otimes b\right)=a\oslash a.
\]

\noindent To prove the first one, we set $u_{1}=a/\sqrt{\left\langle a,a\right\rangle },\ \ v_{1}=b/\sqrt{\left\langle b,b\right\rangle },\ \sigma_{1}=\sqrt{\left\langle a,a\right\rangle \left\langle b,b\right\rangle }$
and note that 
\[
Cv_{1}=ab^{T}b/\sqrt{\left\langle b,b\right\rangle }=a\sqrt{\left\langle b,b\right\rangle }=\sigma_{1}u_{1},
\]

\noindent and similarly 
\[
C^{T}u_{1}=ba^{T}a/\sqrt{\left\langle a,a\right\rangle }=\sigma_{1}v_{1}.
\]

\noindent Hence, it follows from (\ref{eq:eq-8-2}) that $u_{1},v_{1}$
are singular vectors of $C$ corresponding to the singular value $\sigma_{1}$.
Additionally, $u_{1}$ is a positive eigenvector for the eigenvalue
$\sigma^{2}_{1}$ of $CC^{T}$. Thus, by Remark \ref{rem:rem-7-1},
$\sigma^{2}_{1}$ is the largest eigenvalue, and so $\sigma_{1}$
is the largest singular value of $C=a\otimes b$. Now, Definition
\ref{def:def-8-1} can be applied to obtain 
\[
\mathcal{T}_{\mathcal{C}}(a\otimes b)=u_{1}\oslash u_{1}=a\oslash a,
\]

\noindent which completes the proof.
\end{proof}
\begin{theorem}
\noindent\label{thm:thm-8-7}The SVD projection $\mathcal{T}_{\mathcal{C}}$
is a composite projection of the eigenvector and tensor projections
$\mathcal{R}_{\mathcal{C}}$ and $\mathcal{T}_{\mathcal{G}}$: 
\[
\mathcal{T}_{\mathcal{C}}M=\mathcal{R}_{\mathcal{C}}\left(\mathcal{T}_{\mathcal{G}}M\right)=u_{1}\oslash u_{1}
\]

\noindent for every PC matrix $M$ in $\mathbb{R}^{n\times n}$.
\end{theorem}
\begin{proof}
\noindent In view of the Eckart-Young-Mirsky theorem we have 
\[
\mathcal{T}_{\mathcal{G}}M=\left(\sqrt{\sigma_{1}}u_{1}\right)\otimes\left(\sqrt{\sigma_{1}}v_{1}\right).
\]

\noindent Hence, one can apply Lemma \ref{lem:lem-8-6} to finish
the proof.
\end{proof}

\section{Conclusions}

\noindent In this study, we have constructed a first orthogonal basis
for the space ${\mathcal{A}}_{n}$ of all additively consistent matrices
in ${\mathbb{R}}^{n\times n}.$ It has solved the problem stated in
the year 1997 by Koczkodaj and Orłowski in a paper \citep{Koczkodaj1997aobf}
on pairwise comparisons. Recently, this problem has been considered
in \citep{Benitez2024ceof,Koczkodaj2020oopo}, but explicit formulae
for the orthogonal basis of $\mathcal{A}_{n}$ have not been found
until now. The present successful attempt has been made possible due
to a tensor description of the explicit formulae \citep{Koczkodaj2020oopo,Smarzewski2020cpai}
for the best additively consistent approximations of skew-symmetric
matrices.

\noindent The orthogonal basis has shed new light not only on the
orthogonal windowing of PC matrices, but also on the Saaty and SVD
windowing. These three kinds of PC windowing have been discussed both
from the theoretical and computational points of view. As far as we
know, the first two were introduced in \citep{Koczkodaj1997aobf,Saaty1980tahp}
and have been used in pairwise comparisons ever since. The third one
was inspired by an interesting paper of Gass and Rapcsak \citep{Gass2004svdi}.
In the previous section, we showed that the SVD widowing is a unique
method which is both norm and eigenvector based. Therefore, it shares
properties of the windowing of orthogonal and Saaty types.

\noindent Further theoretical studies along these lines are needed,
as well as numerical comparisons of the orthogonal and Saaty windowing
with the SVD windowing of PC matrices in real-world applications.
Of course, such numerical experiments may take into account existing
comparisons of the logarithmic least squares method with the eigenvector
method presented in hundreds of papers and books, for example, in
\citep{Bozoki2008osak,Golany1993ameo,Saaty1980tahp,Saaty1984coel,Sipahi2010tahp}.

\noindent Although it is a common approach in pairwise comparisons
to use only the Frobenius norm \citep{Magnot2023agmf}, it seems to
be too restrictive. In fact, an appropriate choice of the norm is
a very important problem in real-world applications of general approximation
theory, because it may significantly improve or worsen the quality
of approximations, see e.g. Karlin and Studden \citep{Karlin1966tswa},
Laurent \citep{Laurent1975aao}, and Szego \citep{Szego1939op}. It
would be strange if the same were not true in the special case of
consistent approximations in pairwise comparisons. Therefore, further
studies of this problem are needed, at least for the most important
case of the spectral matrix norm \citep{Perron1907ctdm}. Here we
only mention that some changes of norm may significantly change the
consistent approximations \citep{Koczkodaj2020oopo,Smarzewski2020cpai}.

\noindent Finally, we briefly consider the problem of constructing
an additively consistent basis for ${\mathcal{A}}_{n}$, which is
orthogonal with respect to a weighted inner product in ${\mathbb{R}}^{n\times n}$
defined by

\noindent
\[
\left\langle X,Y\right\rangle _{W}=\sum^{n}_{i,j=1}w_{ij}x_{ij}y_{ij}\,\,\,\,\text{for}\,\,X,Y\in\mathbb{R}^{n\times n},
\]

\noindent where ${W=[w}_{ij}]>0.$ In this case, an orthogonal basis
may be computed by applying Gram-Schmidt $\left\langle \cdot,\cdot\right\rangle _{W}$
-- orthogonalization either to the basis $\left\{ A_{1},\dots,A_{n-1}\right\} $
or to the orthogonal basis $\left\{ B_{1},\dots,B_{n-1}\right\} $
presented in Sections \ref{sec:Characterization-of-additive} and
\ref{sec:Orthogonal-basis-of}. The same assertion is evidently true
for any inner product in $\mathbb{R}^{n\times n}$.

\section*{Data Availability Statement}

Data sharing is not applicable to this article as no datasets were
generated or analyzed during the current study.

\section*{Acknowledgment}

Konrad Kułakowski has been supported by the National Science Centre,
Poland within the grant VIRGO 2024/55/B/HS4/00860. 

\bibliographystyle{plain}
\bibliography{papers_biblio_reviewed}

\end{document}